\documentclass[10pt,twocolumn,twoside]{IEEEtran}
\usepackage{cite}
\usepackage{amssymb}
\usepackage{float}
\usepackage{eqparbox}
\usepackage{multirow}

\usepackage[cmex10]{amsmath}
\usepackage{amsmath}
\usepackage{amsthm}
\usepackage{thmtools}
\usepackage{algorithm}
\usepackage{algorithmic}

\usepackage{array}

\usepackage{perpage} %the perpage package
\MakePerPage{footnote} %the perpage package command

\usepackage{mdwmath}
\usepackage{mdwtab}
\usepackage{amsmath}
\usepackage{amsfonts}
\makeatletter \def\amsbb{\use@mathgroup \M@U \symAMSb} \makeatother
\usepackage{bbold}
\usepackage{algorithmic}
\usepackage{algorithm}
\algsetup{linenosize=\small}
\usepackage{setspace} %
\makeatletter \def\amsbb{\use@mathgroup \M@U \symAMSb} \makeatother
\usepackage{bbold}
\usepackage[tight,footnotesize]{subfigure}
\usepackage[belowskip=-8pt,aboveskip=0pt]{caption}
\usepackage{graphicx}
\usepackage{booktabs} %table package
\usepackage{lipsum}
\usepackage{amssymb}
\usepackage{multirow}
\usepackage{stfloats}  %two column-wide tables
\newcommand{\scrsmall}{\scriptscriptstyle }
\def\bphi{\mbox{\boldmath $\phi$}}
\def\bpsi{\mbox{\boldmath $\psi$}}

\def\diag{\mbox{\rm{diag}}}
\def\col{\mbox{\rm{col}}}
\def\msd{{\rm{msd}}}
\def\emse{{\rm{emse}}}

\def\vec{{\rm{vec}}}

\def\Ind{{\rm{Ind}}}

\def\net{{\rm{net}}}
\def\Tr{{\rm{Tr}}}
\def\E{{\amsbb{E}}}
\def\u{\boldsymbol{u}}
\def\n{\boldsymbol{n}}

\def\v{\boldsymbol{v}}

\def\d{\boldsymbol{d}}
\def\z{\boldsymbol{z}}

\def\w{\boldsymbol{w}}
\def\h{\boldsymbol{h}}
\def\p{\boldsymbol{p}}

\def\be{\begin{equation}}
\def\ee{\end{equation}}
\def\ba{\begin{align}}
\def\ea{\end{align}}
\begin{document}
\title{Estimation of Space-Time Varying Parameters Using a Diffusion LMS Algorithm}
\author{Reza~Abdolee$^*$,~\IEEEmembership{Student Member,~IEEE,}
        Benoit~Champagne,~\IEEEmembership{Senior Member,~IEEE,}
        and~ \mbox{Ali~H. Sayed},~\IEEEmembership{Fellow,~IEEE}%
\thanks{Copyright (c) 2013 IEEE. Personal use of this material is permitted. However, permission to use this material for any other purposes must be obtained from the IEEE by sending a request to pubs-permissions@ieee.org. }
\thanks{R. Abdolee and B. Champagne are with the Department
of Electrical and Computer Engineering, McGill University, Montreal,
QC, H3A 2A7 Canada (e-mail: reza.abdolee@mail.mcgill.ca, benoit.champagne@mcgill.ca).}%
\thanks{A. H. Sayed  is with the Department of Electrical Engineering, University
of California, Los Angeles, CA 90095 USA (e-mail: sayed@ee.ucla.edu).}
\thanks{The work of R. Abdolee and B. Champagne was supported by the Natural Sciences and Engineering Research Council (NSERC) of Canada. The work of A. H. Sayed was supported in part by NSF grant CCF-1011918.}
%\date
%\thanks{Manuscript re\bigwedgeceived October 20, 2011; revised January 11, 2007.}}
}
%\markboth{IEEE Transaction on Signal Processing,~Vol.~xx, No.~xx, November~2013} {Abdolee \MakeLowercase{\textit{et al.}}: Estimation of Space-Time Varying Parameters Using a Diffusion LMS Algorithm}
\maketitle
%---------------------------------------------------Abstract---------------------------------------------%
\begin{abstract}
We study the problem of distributed adaptive estimation over networks where nodes cooperate to estimate physical parameters that can vary over both \emph{space} and \emph{time} domains. We use a set of basis functions to characterize the space-varying nature of the parameters and propose a diffusion least mean-squares (LMS) strategy to recover these parameters from successive time measurements. We analyze the stability and convergence of the proposed algorithm, and derive closed-form expressions to predict its learning behavior and steady-state performance in terms of mean-square error. We find that in the estimation of the space-varying parameters using distributed approaches, the covariance matrix of the regression data at each node becomes rank-deficient. Our analysis reveals that the proposed algorithm can overcome this difficulty to a large extent by benefiting from the network  stochastic matrices that are used to combine exchanged information between nodes. We provide computer experiments to illustrate and support the theoretical findings.
\end{abstract}
\IEEEpeerreviewmaketitle
%----------------------------------------------------------------------------------------------------------------------------------------------------------%
%------------------------------------------------------------------------Introduction----------------------------------------------------------------------%
\section{introduction}
\IEEEPARstart{I}{N} previous studies on diffusion algorithms for adaptation over networks, including least-mean-squares (LMS) or recursive least squares (RLS) types, the parameters being estimated are often assumed to be \emph{space-invariant} \cite{cattivelli2008diffusion,cattivelli2010diffusion,chen2012diffusion,chouvardas2011adaptive,tu2012diffusion,sayed2012diffusion}. In other words, all agents are assumed to sense and measure data that arise from an underlying physical model that is represented by fixed parameters over the spatial domain. Some  studies considered particular applications of diffusion strategies to data that arise from potentially different  models \cite{tu2012adaptive,di2012decentralized}. However, the proposed techniques in these works are not  immediately applicable to scenarios where the estimation parameters vary over space across the network. This situation is encountered in many applications, including the monitoring of fluid flow in underground porous media \cite{lee1987estimation}, the tracking of population dispersal in ecology \cite{holmes1994partial}, the localization of distributed sources in dynamic systems \cite{alpay2000model}, and the modeling of diffusion phenomena in inhomogeneous media \cite{van1988diffusion}. In these applications, the space-varying parameters being estimated usually result from discretization of the coefficients of an underlying partial differential equation through spatial sampling.

The estimation of spatially-varying parameters has been addressed in several previous studies, including \cite{chung1988identification,richter1981numerical,isakov2000identification,demetriou2007process, demetriou2009estimation}. In these works and other similar references on the topic, the solutions typically rely on the use of a central processing (fusion) unit and less attention is paid to distributed and in-network processing solutions. Distributed algorithms are useful in large networks when there is no powerful fusion center and when the energy and communication resources of individual nodes are limited. Many different classes of distributed algorithms for parameter estimation over networks have been proposed in the recent literature, including incremental method\cite{bertsekas1997new,Nedia2001Incremental,rabbat2005quantized,lopes2007incremental,li2010distributed}, consensus methods \cite{tsitsiklis1986distributed,xiao2006space,stankovic2007decentralized,braca2008running,sardellitti2010fast,aysal2009broadcast,boyd2006randomized,dimakis2010gossip,srivastava2011distributed,kar2011convergence,di2011bio,hu2010adaptive},
and diffusion methods\cite{lopes2008diffusion,cattivelli2010diffusion,chen2012diffusion,sayed2012diffusion,sayed2013DiffusionMagazine,zhao2012performance}. Incremental techniques require to set-up a cyclic path between nodes over the network and are therefore sensitive to link failures. Consensus techniques require doubly-stochastic combination policies and can cause network instability in applications involving continuous adaptation and tracking \cite{tu2012diffusion}. In comparison, diffusion strategies demonstrate a stable behavior over networks regardless of the topology and endow networks with real-time adaptation and learning abilities \cite{sayed2012diffusion,sayed2013DiffusionMagazine,tu2012diffusion}.

Motivated by these considerations, in this paper, we develop a distributed LMS algorithm of the diffusion type to enable the estimation and tracking of parameters that may vary over  both \emph{space} and \emph{time}. Our approach starts by introducing a linear regression model to characterize space-time varying phenomena over networks. This model is derived by discretizing a representative second-order partial differential equation (PDE), which can be useful in characterizing many dynamic systems with spatially-varying properties.
We then introduce a set of basis functions, e.g., shifted Chebyshev polynomials, to represent the space-varying parameters of the underlying phenomena in terms of a finite set of space-invariant expansion coefficients. Building on this representation, we develop a diffusion LMS strategy that cooperatively  estimates, interpolates, and tracks the model parameters over the network. We analyze the convergence and stability of the developed algorithm, and derive closed-form expressions to characterize the learning and convergence behavior of the nodes in mean-square-error sense. It turns out that in the context of space-time varying models, the covariance matrices of the regression data at the various nodes can become rank deficient. This property influences the learning behavior of the network and causes the estimates to become biased. We elaborate on how the judicious use of stochastic combination matrices can help alleviate this difficulty.

The paper is organized as follows. In Section \ref{sec.:SpaceDependentLinearRegression}, we introduce a space-varying linear regression model which is motivated from a physical phenomenon characterized by a PDE, and formulate an optimization problem to find the unknown parameters of the introduced model. In Section \ref{sec.:AlgorithmDevelopment}, we derive a diffusion LMS algorithm that solves this problem in a distributed and adaptive manner. We analyze the performance of the algorithm in Section \ref{sec.:performance_analysis}, and present the numerical results of computer simulations in Section \ref{sec.:results}. The concluding remarks appear in \mbox{Section \ref{sec.:conclusion}}.

\textit{Notation:}  Matrices are represented by upper-case and vectors by lower-case letters. Boldface fonts are reserved for random variables and normal fonts are used for deterministic quantities. Superscript $(\cdot)^T$ denotes transposition for real-valued vectors and matrices and $(\cdot)^{\ast}$  denotes conjugate transposition for complex-valued vectors and matrices. The symbol $\E[\cdot]$ is the expectation operator, $\text{Tr}(\cdot)$ represents the trace of its matrix argument and diag$\{\cdot\}$ extracts the diagonal entries of a matrix, or constructs a diagonal matrix from a vector. $I_M$ represents the identity matrix of size $M\times M$ (subscript $M$ is omitted when the size can be understood from the context). The vec$(\cdot)$ operator vectorizes a matrix by stacking its columns on top of each other. A set of vectors are stacked into a column vector by $\col\{\cdot\}$.
%------------------------------------------------------------------------------------------------------------------------------------------------------------%
%--------------------------------------------------------------------------Section-2------------------------------------------------------------------------%
\section{Modeling and Problem Formulation}
\label{sec.:SpaceDependentLinearRegression}
In this section, we motivate a linear regression model that can be used to describe dynamic systems with spatially varying properties. We derive the model from a representative second-order one-dimensional PDE that is used to characterize the evolution of the pressure distribution in inhomogeneous media and features a diffusion coefficient and an input source, both of which vary over space. Extension and generalization of the proposed approach, in modeling space-varying phenomena, to  PDEs of higher order or defined over two-dimensional space are generally straightforward (see, e.g., Section \ref{subsec.:Diffusion LMS for Process Estimation}).

The PDE we consider is expressed as \cite{van1988diffusion,mattheij2005partial}:
\begin{align}
\frac{\partial f(x,t)}{\partial t}=\frac{\partial}{\partial x}\Bigg ( \theta(x)\frac{\partial f(x,t)}{\partial x}  \Bigg)+q(x,t)
\label{eq.:one_dimensional_pde1}
\end{align}
where $(x,t)\in [0,L]\times [0,T]$ denote the space and time variables with upper limits $L \in {\amsbb R}^+$ and $T \in {\amsbb R}^+$, respectively,  $f(x,t):{\amsbb R}^2 \rightarrow {\amsbb R}$, represents the system distribution (e.g., pressure or temperature) under study, $\theta(x){\text :}\, {\amsbb R} \rightarrow {\amsbb R}$ is the space-varying diffusion coefficient and $q(x,t){\text :}\, {\amsbb R}^2 \rightarrow {\amsbb R}$  is the input distribution that includes sources and sinks.  PDE (\ref{eq.:one_dimensional_pde1}) is assumed to satisfy the Dirichlet boundary conditions\footnote{Generalization of the boundary conditions to nonzero values is possible as well.}, $f(0,t)=f(L,t) =0$ for all $t \in [0,T]$. The distribution of the system at $t=0$ is given by $f(x,0) =y(x)\; \text{for} \; x\in[0,\,L]$. It is convenient to rewrite (\ref{eq.:one_dimensional_pde1}) as:
\begin{align}
\frac{\partial f(x,t)}{\partial t}= \theta(x)\frac{{\partial}^2 f(x,t)}{{\partial x}^2}+\frac{\partial \theta(x)}{\partial x}\frac{\partial f(x,t)}{\partial x}+q(x,t)
\label{eq.:one_dimensional_pde2}
\end{align}
and employ the finite difference method (FDM) to discretize the PDE over the time and space domains \cite{thomas1995numerical}. For $N$ and $P$ given positive integers, let $\Delta x=L/(N+1)$  and
$x_k=k \Delta x$ for $k \in \{0,1,2,\ldots,N+1\}$, and similarly, let $\Delta t=T/P$ and $t_i=i\Delta t$ for $i \in \{0,1,2, \ldots, P\}$. We further introduce the sampled values of the pressure distribution $ f_k(i)\triangleq f(x_k ,t_i)$, input $q_k(i)\triangleq q(x_k,t_i)$, and  space-varying coefficient $\theta_k\triangleq\theta(x_k)$. It can be verified that applying FDM  to (\ref{eq.:one_dimensional_pde2}), yields the following recursion:
\begin{equation}
 f_k(i)={u}_{k,i}h^o_k+\Delta t \, q_k(i-1), \quad k \in \{1,2,\ldots,N\}
\label{eq.:fieldRecursion}
\end{equation}
where the vectors $h^o_k \in {\amsbb R}^{3\times 1}$ and ${u}_{k,i} \in {\amsbb R}^{1\times 3}$  are defined as
\begin{align}
&{h}^o_k\triangleq[h^o_{1,k},h^o_{2,k},h^o_{3,k}]^T \label{eq.:local_vectors_components}\\
&{u}_{k,i}\triangleq[f_{k-1}(i-1),\,  f_k(i-1), \, f_{k+1}(i-1)]
\end{align}
the entries $h^o_{m,k} \in {\amsbb R}$ are:
\begin{align}
h^o_{1,k}& = \frac{\nu}{4}(\theta_{k-1}+4\theta_k-\theta_{k+1})  \label{eq.:local_vectors_relation_with_theta1} \\
h^o_{2,k}&=1-2\nu \, \theta_k
\label{eq.:local_vectors_relation_with_theta2}\\
h^o_{3,k}& =\frac{\nu}{4}(-\theta_{k-1}+4\theta_k+\theta_{k+1}) \label{eq.:local_vectors_relation_with_theta3}
\end{align}
and $\nu={\Delta t}/{\Delta x^2}$. Note that relation (\ref{eq.:fieldRecursion}) is defined for $k \in \{1,2,\cdots,N\}$, i.e., no data sampling is required to be taken at $x=\{0,L\}$ because $f_0(i)$ and $f_{N+1}(i)$ respectively correspond to the known boundary conditions $f(0,t)$ and $f(L,t)$.
For monitoring purposes (e.g., estimation of $\theta(x)$), sensor nodes collect noisy measurement samples of $f(x,t)$ across the network. We denote these scalar measurement samples by
\begin{equation}
 \z_k(i)=  f_k(i)+ \n_k(i)
\label{eq.:NoisyMeasurement}
\end{equation}
where $\n_k(i) \in {\amsbb R}$ is random noise term.
Substituting (\ref{eq.:fieldRecursion}) into (\ref{eq.:NoisyMeasurement}) leads to
\be
\d_k(i)=u_{k,i}h^o_k+ \n_k(i)
\label{eq.:state_dependent_regression0}
\ee
where
\be
\d_k(i)\triangleq \z_k(i)-\Delta t \,q_k(i-1)
\ee
The space-dependent model (\ref{eq.:state_dependent_regression0}) can be generalized to accommodate higher order PDE's, or to describe systems with more than one spatial dimension. In the generalized form, we assume that $u_{k,i}$ is random due to the possibility of sampling errors, and therefore represent it using boldface notation $\u_{k,i}$. We also let $h_k^o$ and $\u_{k,i}$ be $M$-dimensional vectors. In addition, we denote the noise more generally by the symbol
$\v_k(i)$ to account for different sources of errors, including
the measurement noise shown in (\ref{eq.:NoisyMeasurement}) and modeling errors. Considering this generalization, the space-varying regression model that we shall consider is of the form:
\begin{equation}
 \d_k(i)=\u_{k,i}h^o_k+ \v_k(i)
\label{eq.:state_dependent_regression}
\end{equation}
where $\d_k(i) \in {\amsbb R}, \u_{k,i} \in {\amsbb R}^{1 \times M}, h^o_k \in {\amsbb R}^{M \times 1}$ and $\v_k(i) \in {\amsbb R}$.
 In this work, we study networks that monitor phenomena characterized by regression models of the form
(\ref{eq.:state_dependent_regression}), where the objective is to estimate the space-varying parameter vectors $h_k^o$ for  $k \in \{1,2,\cdots,N\}$. In particular, we seek a distributed solution in the form of an adaptive algorithm with a diffusion mode of cooperation to enable the nodes to estimate and track these parameters over both space and time. The available information for estimation of the $\{h^o_k\}$ are the measurement samples, $\{\d_k(i), \u_{k,i}\}$, collected at the $N$ spatial position $x_k$, which we take to represent  $N$ nodes.

Several studies, e.g., \cite{chung1988identification,richter1981numerical,isakov2000identification}, solved space-varying parameter estimation problems using {\em \mbox{non-adaptive} centralized} techniques.
In centralized optimization, the space-varying parameters $\{h^o_k\}$ are found by minimizing the following global cost function over the variables $\{h_k\}$:
\vspace{-0.2cm}
\begin{equation}
J(h_1,\ldots,h_N)\triangleq \sum_{k=1}^N J_k(h_k)
\label{eq.:DecomposedCostFunction}
\end{equation}
where
\begin{equation}
J_k(h_k) \triangleq \E| \d_k(i)-\u_{k,i}h_k|^2
\label{eq.:localCostFunction}
\end{equation}
To find $h^o_k$ using distributed mechanisms, however, preliminary steps are required to transform the global cost (\ref{eq.:DecomposedCostFunction}) into a suitable form convenient for decentralized optimization \cite{cattivelli2010diffusion}. Observe from (\ref{eq.:local_vectors_relation_with_theta1})-(\ref{eq.:local_vectors_relation_with_theta3}) that collaborative processing is beneficial in this case because the $h_k^o$ of neighboring nodes are related to each other through the space-dependent function $\theta(x)$.
\theoremstyle{remark}
\newtheorem{remark}{Remark}
\begin{remark}
Note that if nodes individually estimate their own space-varying parameters by minimizing $J_k(h_k)$, then at each time instant, they will need to transmit their estimates to a fusion center for interpolation, in order to determine the value of the model parameters over regions of space where no measurements were collected. Using the proposed distributed algorithm in Section \ref{subsec.:diffusion_strategy}, it will be possible to update the estimates and interpolate the results in a fully distributed manner. Cooperation also helps the nodes refine their estimates and perform more accurate interpolation. \hfill $\blacksquare$
\end{remark}
%------------------------------------------------------------------------------------------------------------------------------------------------------------%
%--------------------------------------------------------------------------Section-3------------------------------------------------------------------------%
\section{Adaptive Distributed Optimization}
\label{sec.:AlgorithmDevelopment}
In distributed optimization over networked systems, nodes achieve their common objective through collaboration. Such an objective may be defined as finding a global parameter vector that minimizes a given cost function that encompasses the entire set of nodes. For the problem stated in this study, the unknown parameters in (\ref{eq.:DecomposedCostFunction}) are node-dependent. However, as we explained in Section \ref{sec.:SpaceDependentLinearRegression}, these space-varying parameters are related through a well-defined function, e.g., $\theta(x)$ over the spatial domain. In the continuous space domain, the entries of each $h^o_k$, i.e., $\{h^o_{1,k},\cdots,h^o_{M,k}\}$ can be interpreted as samples of $M$ unknown space-varying parameter functions $\{h^o_1(x),\cdots,h^o_M(x)\}$ at location $x=x_k$, as illustrated in \mbox{Fig. \ref{fig.:network-topology-and-parameters}}.
\begin{figure}[!h]
\centering
\includegraphics[width=7.5 cm,height=5cm]{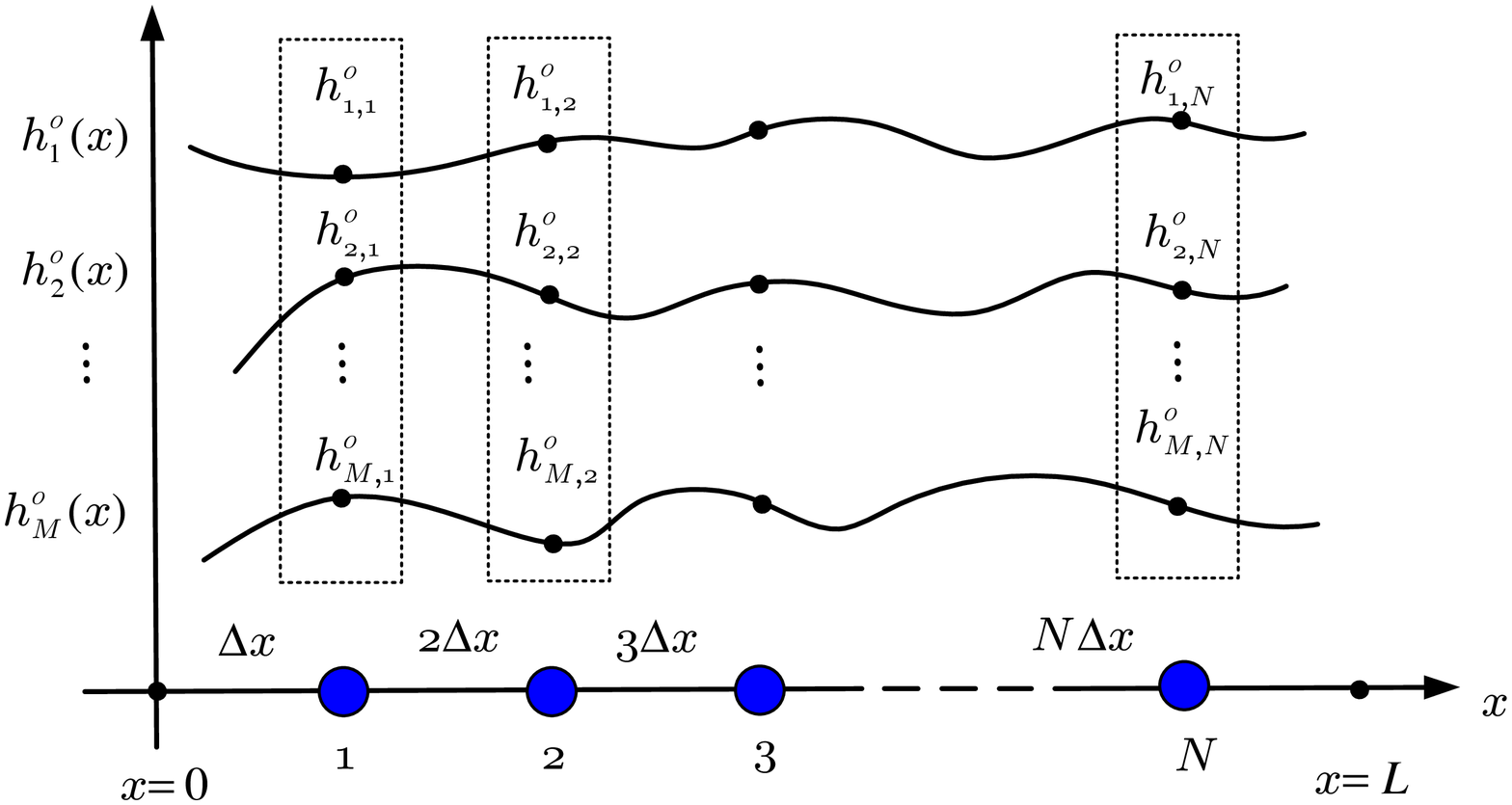}
\caption{\small {An example of the space-varying parameter estimation problem over a one-dimensional network topology. The larger circles on the $x$-axis represent the node locations at \mbox{$x=x_k$}. These nodes collect samples $\{\d_k(i),\u_{k,i}\}$ to estimate the space-varying parameters $\{h^o_k\}$}. For simplicity in defining the vectors $b_k$ in (\ref{eq.:b_k-entries}), for this example, we assume that the node positions $x_k$  are uniformly spaced, however, generalization to non-uniform spacing is straightforward.}
\label{fig.:network-topology-and-parameters}
\end{figure}

We can now use the well-established theory of interpolation  to find a set of linear expansion coefficients, common to all the nodes, in order to estimate  space-varying parameters using distributed optimization. Specifically, we assume that the \mbox{$m$-th} unknown space-varying parameter function, $h^o_m(x)$ can be expressed as a unique linear combination of some $N_b$ space basis functions, i.e.,
\begin{equation}
h^o_m(x)=W_{m,1}b_1(x)+W_{m,2}b_2(x)+\cdots+W_{m,N_b}b_{N_b}(x)
\label{eq.:countinoush(x)}
\end{equation}
where $\{W_{m,n}\}$ are the unique expansion coefficients and $\{b_{n}(x)\}$ are the basis functions. In the application examples treated in Section \ref{sec.:results}, we adopt shifted Chebyshev polynomials as basis functions,
which are generated using the following expressions \cite{mason2003chebyshev}
\begin{align}
&b_{1}(x)=1, \qquad b_{2}(x)=2x-1  \\
&b_{n+1}(x)=2(2x-1) b_{n}(x)-b_{n-1}(x),\quad 2<n<N_b
\label{eq.:ShiftedCountinuesChebyshev}
\end{align}
The choice of a suitable set of basis functions $\{b_n(x)\}_{n=1}^{N_b}$ is application-specific and guided by multiple considerations such as representation efficiency, low computational complexity, interpolation capability, and other desirable properties, such as orthogonality. Chebyshev basis functions yield good results in terms of the above criteria and helps avoid the Runge's phenomenon at the endpoints of the space interval \cite{mason2003chebyshev}.

The sampled version of the \mbox{$m$-th} space-varying parameter $h^o_m(x)$ in (\ref{eq.:countinoush(x)}),  at $x=x_k=k\Delta x$, can be written as:
\begin{equation}
h^o_{m,k}=W_m^T b_k
\label{eq.:m-th-entryOfSpaceVarPara}
\end{equation}
where
\vspace{-0.2cm}
\begin{align}
W_m&\triangleq [W_{m,1},W_{m,2},\cdots,W_{m,N_b}]^T \\
b_{k}&\triangleq [b_{1,k},\cdots,b_{N_b,k}]^T
\label{eq.:b_k-entries}
\end{align}
and each entry $b_{n,k}$ is obtained by sampling the corresponding basis function at the same location, i.e.,
\be
b_{n,k} \triangleq  b_n(x_k) =b_n(k\Delta x)
\label{eq.:b-nk}
\ee
Collecting the sampled version of all $M$ functions $ h_m^o(x)$ for $m\in \{1,\cdots,M\}$  into a column vector as
\be
h^o_k=[h^o_{1,k},h^o_{2,k},\cdots, h^o_{M,k}]^T
\ee
and using (\ref{eq.:m-th-entryOfSpaceVarPara}), we arrive at:
\begin{equation}
h^o_k=W^o b_k
\label{eq.:B-spline function}
\end{equation}
where
\small
\begin{equation}
W^o\triangleq \left [ \begin{array}{cccc}
W^o_{1,1} & W^o_{1,2} & \ldots & W^o_{1,N_b}\\
W^o_{2,1} & W^o_{2,2} & \ldots & W^o_{2,N_b}\\
\vdots & \vdots & \ldots & \vdots\\
W^o_{M,1} & W^o_{M,2} & \ldots & W^o_{M,N_b}
\end{array}\right]
\end{equation}
\normalsize
%------------------------------------------------------------------------------%
\theoremstyle{remark}
\begin{remark}
Several other interpolation techniques can be used to obtain the basis functions $b_n(x)$, such as the so-called kriging method \cite{isaaks1989introduction}. The latter is a data-based weighting approach, rather than a distance-based interpolation. It is applicable in scenarios where the unknown random field to be interpolated, in our case $h^o_k$, is wide-sense stationary; accordingly, it requires knowledge about the means and covariances of the random field over space, as employed in \cite{kim2011cooperative}. If these covariances are not available, then the variogram models, describing the degree of spatial dependence of the random field, are used to generate substitutes for these covariances \cite{cressie1993statistics}. However, {\it a-priori} knowledge about the parameters of variogram models, including nugget, sill, and range, are required to obtain the spatial covariances. In this work, since neither the means and covariances nor the parameters of the variogram models of the random fields are available, we focus on interpolation techniques that rely on distance information rather than the statistics of the random field to be interpolated.
\hfill $\blacksquare$
\end{remark}
%--------------------------------------------------------------------------------%
Returning to equation (\ref{eq.:B-spline function}), it is convenient to rearrange $W^o$ into an $MN_b\times 1$ column vector $w^o$ by stacking up the columns of ${W^o}^T$, i.e., $w^o=\vec({W^o}^T)$, and defining the block diagonal matrix $B_k \in {\amsbb R}^{M\times MN_b}$ as:
\begin{equation}
B_k\triangleq  I_M \otimes b^T_k
\label{eq.:B_k}
\end{equation}
Then, relation (\ref{eq.:B-spline function}) can be rewritten in terms of the unique parameter vector $w^o$ as:
\begin{equation}
h^o_k=B_k w^o
\label{eq.:local_global_relation}
\end{equation}
so that substituting $h^o_k$ from (\ref{eq.:local_global_relation}) into (\ref{eq.:state_dependent_regression}) yields:
\begin{align}
&\d_k(i)= \u_{k,i}B_k w^o+ \v_k(i) \label{eq.:measurement_linear_model2}
\end{align}
Subsequently, the global cost function (\ref{eq.:DecomposedCostFunction}) becomes:
\begin{align}
&J(w)=\sum_{k=1}^N \E|\d_k(i)- \u_{k,i}B_k w|^2  \label{eq.:LocalObjectiveFunction}
\end{align}
In the following, we elaborate on how the parameter vector $w^o$ and, hence, the $\{h_{k}^o\}$ can be estimated from the data $\{\d_k(i),\u_{k,i}\}$ using centralized and distributed adaptive optimization.
%---------------------------------------------------------------------%
\subsection{Centralized Adaptive Solution}
\label{sebsec.:CentralizedSolution}
We begin by stating the assumed statistical conditions on the data over the network.
\newtheorem{assm}{Assumption}
\begin{assm}
\label{assm.:regressor assumption}
We assume that  $\{\d_k(i), \u_{k,i}, \v_{k}(i)\}$ in model (\ref{eq.:measurement_linear_model2}) satisfy the following conditions:
\begin{enumerate}
\item $\d_k(i)$ and $ \u_{k,i}$ are zero-mean, jointly wide-sense stationary random processes with second-order moments:
\begin{align}
r_{du,k}&=\E[\d_k(i)\u_{k,i}^T] \in {\amsbb R}^{M \times 1}\\
R_{u,k}&=\E[\u_{k,i}^T\u_{k,i}] \in {\amsbb R}^{M \times M}
\end{align}
\item The regression data $\{ \u_{k,i}\}$  are i.i.d. over time, independent over space, and their covariance matrices, $R_{u,k}$, are positive definite for all $k$.
\item The noise processes $\{\v_k(i)\}$ are zero-mean, i.i.d. over time, and independent over space with variances $\{\sigma^2_{v,k}\}$.
\item The noise process $\v_k(i)$ is independent of the regression data  $\u_{m,j}$ for all $i,j$ and $k,m$.
\hfill $\blacksquare$
\end{enumerate}
\end{assm}
The optimal parameter $w^o$ that minimizes (\ref{eq.:LocalObjectiveFunction}) can be found by setting the gradient vector of  $J(w)$ to zero.
This yields the following normal equations:
\be
\Big(\sum_{k=1}^N \bar{R}_{u,k}\Big) w^o = \sum_{k=1}^N \bar{r}_{du,k}
\label{eq.CentralizedNormalEquation}
\ee
where $\{\bar{R}_{u,k}, \bar{r}_{du,k}\}$ denote the second-order moments of $\u_{k,i} B_{k}$ and $\d_k(i)$:
\be
\bar{R}_{u,k}\triangleq B_k^T R_{u,k} B_k,\;\;\;\;
\bar{r}_{du,k} \triangleq B_k^T r_{du,k}
\label{eq.:bar_Ru}
\ee
It is clear from (\ref{eq.CentralizedNormalEquation}) that when $\sum_{k=1}^N \bar{R}_{u,k}>0$, then $w^o$ can be determined uniquely. If, on the other hand, $\sum_{k=1}^N \bar{R}_{u,k}$ is singular, then we can use its pseudo-inverse to recover the minimum-norm solution of (\ref{eq.CentralizedNormalEquation}). Once the global solution is estimated, we can retrieve the space-varying parameter vectors $h^o_k$ by substituting $w^o$ into (\ref{eq.:local_global_relation}).

Alternatively the solution $w^o$ of (\ref{eq.CentralizedNormalEquation}) can be sought iteratively by using the following steepest descent recursion:
\begin{align}
& \w^{(c)}_i= \w^{(c)}_{i-1}+\mu \sum_{k=1}^N \big (\bar{r}_{du,k}-\bar{R}_{u,k} \w^{(c)}_{i-1}\big )\label{eqe:centralized_steepest_eq1}
\end{align}
where $\mu>0$ is a step-size parameter and $ \w^{(c)}_i$ is the estimate of $w^o$ at the $i$-th iteration. Recursion (\ref{eqe:centralized_steepest_eq1}) requires the centralized processor to have knowledge of the covariance matrices, $\bar{R}_{u,k}$, and cross covariance vectors, $\bar{r}_{du,k}$, across all nodes.
In practice, these moments are unknown in advance, and we therefore use their instantaneous approximations in (\ref{eqe:centralized_steepest_eq1}). This substitution leads to the centralized LMS strategy (\ref{eq.:centralized-step1})--(\ref{eq.:centralized-step2}) for space-varying parameter estimation over networks.
\begin{algorithm}
\caption{:\;Centralized LMS}
\label{eq.CentralizedLmsSpaVarParaEst}
\begin{align}
&\w^{(c)}_i= \w^{(c)}_{i-1}+\mu\displaystyle\sum_{k=1}^N B_k^T \u_{k,i}^T \big(\d_k(i)-\u_{k,i} B_k \w^{(c)}_{i-1}\big)
\label{eq.:centralized-step1}\\
&\h_{k,i}=B_k \w^{(c)}_i, \quad k \in \{1,2,\cdots,N\}
\label{eq.:centralized-step2}
\end{align}
\end{algorithm}
\par \noindent
In this algorithm, at any given time instant $i$, each node transmits its data $\{\u_{k,i},\d_k(i)\}$ to the central processing unit to update $\w^{(c)}_{i-1}$. Subsequently, the algorithm obtains an estimate for the space-varying parameters, $\h_{k,i}$, by  using the updated estimate $\w^{(c)}_i$, and the basis function matrix at location $k$, (i.e., $B_k$). This latter step can also be used as an interpolation mechanism to estimate the space-varying parameters at locations other than the pre-determined locations $\{x_k\}$, by using the corresponding matrix $B(x)$ for some desired location $x$.
%-------------------------------------------------------------------------------------------%
\subsection{Adaptive Diffusion Strategy}
\label{subsec.:diffusion_strategy}
There are different distributed optimization techniques that can be applied to (\ref{eq.:LocalObjectiveFunction})
in order to estimate  $w^o$ and consequently obtain the optimal space-varying parameters $h^o_k$.
Let $\mathcal{N}_k$ denote the index set of the neighbors of node $k$, i.e., the nodes with which node $k$ can share information (including $k$ itself). One possible optimization strategy is to decouple the global cost (\ref{eq.:LocalObjectiveFunction}) and write it as a set of constrained optimization problems with local variables $w_k$, \cite{boyd2011a}, i.e.,
\begin{align}
&\min_{w_k}\sum_{\ell\in {\cal N}_k} c_{\ell,k}\E| \d_{\ell}(i)- \u_{\ell,i}B_k w_k|^2 \nonumber \\
             &\text{subject to} \quad w_k=w
             \label{minimizationProblemLagrangien}
\end{align}
where $c_{\ell, k}$ are nonnegative entries of a right-stochastic matrix $C \in \amsbb{R}^{N \times N}$ satisfying:
\vspace{-0.2cm}
\begin{align}
&c_{\ell,k}=0 \;\, \text{if} \;\, \ell\notin \mathcal{N}_k \;\;\, \textrm{and} \;\;\; C{\mathbb 1}={\mathbb 1}
\label{eq.sthocastic_matrix_c_conditions}
\end{align}
and ${\mathbb 1}$ is the column vector with unit entries.

The optimization problem (\ref{minimizationProblemLagrangien}) can be solved using, for example, the alternating directions method of multipliers (ADMM) \cite{bertsekas1989parallel,boyd2011a}. In the algorithm derived using this method, the Lagrangian multipliers associated with the constraints need to be updated at every iteration during the optimization process. To this end, information about the network topology is required to establish a hierarchical communication structure between nodes. In addition, the constraints imposed by (\ref{minimizationProblemLagrangien}) require all agents to agree on an exact solution; this requirement degrades  the learning and tracking abilities of the nodes over the network. When some nodes observe relevant data, it is advantageous for them to be able to respond quickly to the data without being critically constrained
by perfect agreement at that stage. Doing so, would enable information to diffuse more rapidly across the network.

A technique that does not suffer from these difficulties and endows networks with adaptation and learning abilities in real-time is the diffusion strategy \cite{lopes2008diffusion,cattivelli2010diffusion,chen2012diffusion,sayed2012diffusion,sayed2013DiffusionMagazine}.
In this technique, minimizing the global cost (\ref{eq.:LocalObjectiveFunction}) motivates solving the following unconstrained local optimization problem for \mbox{$k\in\{1,\cdots,N\}$} \cite{cattivelli2010diffusion}:
\begin{align}
\min_w \Big(&\sum_{\ell\in {\cal N}_k} c_{\ell,k}\E| \d_{\ell}(i)- \u_{\ell,i}B_k w|^2 \nonumber \\
&\hspace{1cm}+\sum_{\ell\in {\cal N}_k \backslash\{k\}} p_{\ell,k}\|w-\psi_{\ell}\|^2\Big) \label{eq.:constrained-local-objectiveFunc.}
\end{align}
where $\psi_{\ell}$ is the available estimate of the global parameter at node $\ell$, ${\cal N}_k\backslash \{k\}$ denotes set ${\cal N}_k$ excluding node $k$, and $\{p_{\ell, k}\}$ are nonnegative scaling parameters. Following the arguments in \cite{cattivelli2010diffusion,chen2012diffusion,sayed2012diffusion}, the minimization of (\ref{eq.:constrained-local-objectiveFunc.}) leads to a general form of the diffusion strategy described by (\ref{eq.diff-step1})--(\ref{eq.diff-step4}), which can be specialized to several simpler yet useful forms.
\begin{algorithm}
\caption{:\;Diffusion LMS}
\label{eq.GeneralizedLmsSpaVarParaEst}
\begin{align}
&\bphi_{k,i-1}=\displaystyle\sum_{\ell\in \mathcal{N}_k} a^{(1)}_{\ell,k} \w_{\ell,i-1} \label{eq.diff-step1}\\
&{\bpsi}_{k,i}=\bphi_{k,i-1}+\mu_k\displaystyle\sum_{\ell\in \mathcal{N}_k}c_{\ell,k} B_{\ell}^T \u_{\ell,i}^T \big(\d_{\ell}(i)- \u_{\ell,i}B_{\ell}\bphi_{k,i-1}\big)\label{eq.diff-step2}\\
&\w_{k,i}=\displaystyle\sum_{\ell\in \mathcal{N}_k}a^{(2)}_{\ell, k}{\bpsi}_{\ell,i} \label{eq.diff-step3}\\
&\h_{k,i}=B_k \w_{k,i} \label{eq.diff-step4}
\end{align}
\end{algorithm}
%
%\vspace{-0.2cm}
\par \noindent
In this algorithm, $\mu_k>0$ is the step-size at node $k$, $\{\w_{k,i}, \bpsi_{k,i}, \bphi_{k,i-1}\}$ are intermediate estimates of $w^o$, $\h_{k,i}$ is an intermediate estimate of $h^o_k$, and  $\{ a^{(1)}_{\ell,k},a^{(2)}_{\ell,k}\}$ are nonnegative entries of left-stochastic matrices $A_1, A_2 \in \amsbb{R}^{N \times N}$
that satisfy:
\begin{align}
& a^{(1)}_{\ell,k}=a^{(2)}_{\ell,k}=0 \quad \text{if} \, \ell\notin \mathcal{N}_k
\label{eq.sthocastic_matrix_a1a2_conditions1} \\
&\quad A_1^T{\mathbb 1}={\mathbb 1}\quad  A_2^T{\mathbb 1}={\mathbb 1}
\label{eq.sthocastic_matrix_a1a2_conditions2}
\end{align}
Each node $k$ in the first combination step  fuses  $\{\w_{\ell,i-1}\}_{\ell \in {\cal N}_k}$ in a convex manner to generate $\bphi_{k,i-1}$. In the following step, named adaptation, each node $k$ uses its own data and that of neighboring nodes, i.e., $\big \{ \u_{\ell,i}, \d_{\ell}(i)\big\}_{\ell \in {\cal N}_k}$  to adaptively update $\bphi_{k,i-1}$ to an intermediate estimate $\bpsi_{k,i}$. In the third step, which is also a  combination, the intermediate estimates $\{ \bpsi_{\ell,i}\}_{\ell \in \mathcal{N}_k}$ are fused to further align the global parameter estimate at node $k$ to that of its neighbors. Subsequently, the desired space-varying parameter $\h_{k,i}$ is obtained from $\w_{k,i}$. Note that each step in the algorithm runs concurrently over the network.
%----------------------------------------------%
\theoremstyle{remark}
\begin{remark}
The main difference between Algorithm \ref{eq.GeneralizedLmsSpaVarParaEst} and the previously developed diffusion LMS strategies in, e.g.,  \cite{lopes2008diffusion,cattivelli2010diffusion,sayed2012diffusion} is in the transformed domain regression data $\u_{\ell,i} B_{\ell}$ in (\ref{eq.diff-step2}) which now have singular covariance matrices. Moreover, there is an additional interpolation step (\ref{eq.diff-step4}). \hfill $\blacksquare$
\end{remark}

%----------------------------------------------%
\theoremstyle{remark}
\begin{remark}
\label{Re.:MNandMNbRelation}
The proposed diffusion LMS algorithm estimates $N M$ spatially dependent variables  $\{h_k^o\}$  using $N_b M$ global invariant coefficients in $w^o$. From the computational complexity and energy efficiency point of view, it seems this is advantageous when the number of nodes, $N$, is greater than the number of basis functions $N_b$. However, even if this is not the case, using the estimated $N_b M$ global invariant coefficients, the algorithm not only can estimate the space-varying parameters at the locations of the $N$ nodes,
but can also estimate the space-varying parameters at locations where no measurements are available. Therefore, even when $N<N_b$, the algorithm is still useful as it can perform interpolation.
\hfill $\blacksquare$
\end{remark}

There are different choices for the combination matrices $\{A_1,A_2,C\}$. For example, the choice $A_1=A_2=C=I$ reduces the above diffusion algorithm to the non-cooperative case where each node runs an individual LMS filter without coordination with its neighbors. Selecting $C=I$ simplifies the adaptation step (\ref{eq.diff-step2}) to the case where node $k$ uses only its own data $\{\d_k(i),\u_{k,i}\}$ to perform local adaptation. Choosing $A_1=I$ and $A_2=A$, for some left-stochastic matrix $A$, removes  the first combination step and the algorithm reduces to an adaptation step followed by combination (this variant of the algorithm has the Adapt-then-Combine or ATC diffusion structure) \cite{cattivelli2010diffusion,sayed2012diffusion}. Likewise, choosing $A_1=A$ and $A_2=I$ removes the second combination step and the algorithm reduces to a combination step followed by adaptation (this variant has the Combine-then-Adapt (CTA) structure of diffusion \cite{cattivelli2010diffusion,sayed2012diffusion}). Often in practice, either the ATC or CTA version of the algorithm is used with $C$  set to $C=I$ such as using the following ATC diffusion version described by equations (\ref{eq.atc-step1})--(\ref{eq.atc-step3}).
%---------Diffusion ATC--------------------%
\begin{algorithm}
\caption{:\;Diffusion ATC}
\label{eq.ATCLmsSpaVarParaEst}
\begin{align}
&{\bpsi}_{k,i}=\w_{k,i-1}+\mu_k B_{k}^T \u_{k,i}^T\big(\d_{k}(i)-\u_{k,i}B_{k}\w_{k,i-1}\big)
\label{eq.atc-step1}\\
&\w_{k,i}=\displaystyle\sum_{\ell\in \mathcal{N}_k}a_{\ell, k}{\bpsi}_{\ell,i}
\label{eq.atc-step2}\\
&\h_{k,i}=B_k \w_{k,i} \label{eq.atc-step3}
\end{align}
\end{algorithm}
%------------------------------------------%

Nevertheless for generality, we shall study the performance of  Algorithm \ref{eq.GeneralizedLmsSpaVarParaEst} for arbitrary matrices $\{A_1,A_2,C\}$ with $C$ right-stochastic and $\{A_1,A_2\}$ left-stochastic. The results can then be specialized to various situations of interest, including ATC, CTA, and the non-cooperative case.

The combination matrices $\{A_1,A_2,C\}$ are normally obtained using some well-known available combination rules such as the Metropolis or uniform combination rules \cite{xiao2006space,lopes2008diffusion,cattivelli2010diffusion}. These matrices can also be treated as free variables in the optimization procedure and used to further enhance the performance of the diffusion strategies. Depending on the network topology and the quality of the communication links between nodes, the optimized values of the combination matrices differ from one case to another\cite{sayed2012diffusion,zhao2012imperfect,takahashi2010diffusion,abdolee2011diffusion}.
%------------------------------------------------------------------------------------------------------------------------------------------------------------%
%--------------------------------------------------------------------------Section-4------------------------------------------------------------------------%
\section{Performance Analysis}
\label{sec.:performance_analysis}
In this section, we analyze the performance of the diffusion strategy (\ref{eq.diff-step1})-(\ref{eq.diff-step4}) in the mean and mean-square sense and derive expressions to characterize the network
mean-square deviation (MSD) and excess mean-square error (EMSE).
In the analysis, we need to consider the fact that the covariance matrices $\{\bar{R}_{u,k}\}_{k=1}^N$ defined in (\ref{eq.:bar_Ru}) are now rank-deficient since we have $N_b>1$. We explain in the sequel the ramifications that follow from this rank-deficiency.
%-----------------------Subsection----------------------------------------%
\subsection{Mean Convergence}
We introduce the local weight-error vectors
\be
\tilde  \w_{k,i}\triangleq w^o- \w_{k,i},\, \tilde{\bpsi}_{k,i} \triangleq w^o-\bpsi_{k,i},\,\tilde{\bphi}_{k,i}\triangleq w^o-\bphi_{k,i}
\ee
and define the network error vectors:
\begin{align}
&\tilde{\bphi}_i \triangleq \col\{\tilde{\bphi}_{1,i},\ldots,\tilde{\bphi}_{N,i}\}\\
&\tilde{\bpsi}_i \triangleq \col\{\tilde{\bpsi}_{1,i},\ldots,\tilde{\bpsi}_{N,i}\}\\
&\tilde  \w_i \triangleq \col\{\tilde \w_{1,i},\ldots,\tilde  \w_{N,i}\}
\end{align}
We collect the estimates from across the network into the block vector:
\begin{align}
&\w_i\triangleq \col\{ \w_{1,i},\ldots, \w_{N,i}\}
\label{eq.:bomegai}
\end{align}
and introduce the following extended combination matrices:
\begin{align}
{\mathcal A_1}& \triangleq A_1\otimes I_{MN_b} \\
{\mathcal A_2}& \triangleq A_2\otimes I_{MN_b} \\
{\mathcal C}& \triangleq C\otimes I_{MN_b}
\end{align}
We further define the block diagonal matrices and vectors:
\begin{align}
\boldsymbol{\cal R}_i& \triangleq \text{diag}\Big \{ \small \sum_{\ell \in {\cal N}_k} c_{\ell, k} B_{\ell}^T \u_{\ell,i}^T \u_{\ell,i} B_{\ell}
: k=1,\cdots,N \normalsize \Big \}\\
{\mathcal M}& \triangleq \text{diag}\big \{ \mu_1 I_{MN_b},\ldots, \mu_N I_{MN_b} \big \} \label{eq.:calM}\\
\boldsymbol{t}_i& \triangleq \col\Big \{\sum_{\ell \in {\cal N}_k} c_{\ell, k} B_{\ell}^T \u_{\ell,i}^T \d_{\ell}(i): k=1,\cdots,N \Big \}\\
\boldsymbol{g}_i& \triangleq {\mathcal C}^T \, \col\big \{B_1^T \u_{1,i}^T \v_1(i),\cdots, B_N^T \u_{N,i}^T \v_N(i) \big \}
\end{align}
and introduce the expected values of $\boldsymbol{\cal R}_i$ and $\boldsymbol{t}_i$:
\begin{align}
&{\cal R}\triangleq \E[\boldsymbol{\cal R}_i] =\text{diag}\big \{{R}_1,\cdots,{R}_N\big\}
\label{eq.:Rcal_rank_defficientDiff}\\
&r \triangleq \E[\boldsymbol{{t}}_i]=\col\big \{r_1,\cdots, r_N\big\}
\label{eq.:Rdcal_rank_defficientDiff}
\end{align}
where
\be
R_k \triangleq \sum_{\ell \in {\cal N}_k} c_{\ell,k} \,\bar{R}_{u,\ell}
\label{eq.:localR}
\ee
\be
r_{k}\triangleq \sum_{\ell \in {\cal N}_k} c_{\ell,k} \,\bar{r}_{du,\ell}
\label{eq.:localr}
\ee
We also introduce an indicator matrix operator, denoted by $\Ind(\cdot)$, such that for any real-valued matrix $X$ with $(k,j)$-th entry $X_{k,j}$, the corresponding entry of $Y=\Ind(X)$ is:
\begin{equation}
Y_{k,j}=\left\{
\begin{array}{l l}
1, \qquad \text{if}\;X_{k,j}>0\\
0,  \qquad \text{otherwise}
\end{array} \right.
\end{equation}
Now from (\ref{eq.diff-step1})--(\ref{eq.diff-step3}), we obtain:
\begin{align}
&\w_i= \boldsymbol{\mathcal{B}}_i \w_{i-1}+{\mathcal A}^T_2{\mathcal M}\boldsymbol{t}_i
\label{eq.:network_vector_update}
\end{align}
where
\be
\boldsymbol{\mathcal{B}}_i \triangleq {\mathcal A}_2^T(I-{\mathcal M}\boldsymbol{\cal R}_i){\mathcal A}^T_1
\ee
In turn, making use of (\ref{eq.:measurement_linear_model2}) in (\ref{eq.:network_vector_update}), we can verify that the network error vector follows the recursion
\begin{align}
\tilde  \w_i=\boldsymbol{\mathcal{B}}_i \tilde  \w_{i-1}-{\mathcal A}_2^T{\mathcal M} \boldsymbol{g}_i
\label{eq.:global_error_vector_w}
\end{align}
By taking the expectation of both sides of (\ref{eq.:global_error_vector_w}) and using \mbox{Assumption \ref{assm.:regressor assumption}}, we arrive at:
\begin{align}
\E[\tilde  \w_i]={\cal B}\, \E[\tilde  \w_{i-1}]
\label{eq.:mean_perfomance}
\end{align}
where in this relation:
\begin{align}
{\cal B}\triangleq \E[\boldsymbol{\mathcal{B}}_i]={\mathcal A}_2^T(I-{\mathcal M}{\cal R}){\mathcal A}^T_1
\label{eq.:calB}
\end{align}
To obtain (\ref{eq.:mean_perfomance}), we used the fact that the expectation of the second term in (\ref{eq.:global_error_vector_w}), i.e., $\E[{\mathcal A}^T_2{\mathcal M} \boldsymbol{g}_i]$, is zero because $ \v_k(i)$ is independent of $ \u_{k,i}$ and $\E[ \v_k(i)]=0$.
The rank-deficient matrices  $\{\bar{R}_{u,k}\}$ appear inside ${\cal R}$ in (\ref{eq.:calB}). We now verify that despite having rank-deficient matrix ${\cal R}$, recursion (\ref{eq.:mean_perfomance}) still guarantees a bounded mean error vector in steady-state.

To proceed, we introduce the eigendecomposition:
\begin{equation}
R_k=Q_k\Lambda_k Q_k^T
\end{equation}
where $Q_k=[q_{k,1},\cdots,q_{k,MN_b}]$ is a unitary matrix with column eigenvectors $q_{k,j}$  and $\Lambda_k=\diag\{\lambda_k(1),\cdots,\lambda_k({MN_b})\}$ is a diagonal matrix with eigenvalues $\lambda_k(j)\geq 0$. For this decomposition, we assume that the eigenvalues of $R_k$ are arranged in descending order, i.e, $\lambda_{\max}(R_k) \triangleq \lambda_k(1) \geq\lambda_k(2)\geq\cdots\geq\lambda_k(MN_b)$, and the rank of $R_k$ is  $L_k\leq MN_b$. If we define ${\cal Q}\triangleq \text{diag}\{Q_1,\ldots
,Q_N\}$ and  $\Lambda\triangleq \diag \{\Lambda_1,\cdots ,\Lambda_{N}\}$, then the network covariance matrix, ${\cal R}$, given by (\ref{eq.:Rcal_rank_defficientDiff}) can be expressed as:
\begin{equation}
{\cal R}={\cal Q}\Lambda {\cal Q}^T
\label{eq.:EigenDecompositionOfD}
\end{equation}
We now note that the  mean estimate vector, $\E[\tilde{\w}_i]$, expressed by (\ref{eq.:mean_perfomance}) will be asymptotically unbiased if the spectral radius of ${\cal B}$, denoted by $\rho({\cal B})$, is strictly less than one. Let us examine under what conditions this requirement is satisfied. Since $A_1$ and $A_2$ are left-stochastic matrices and ${\cal R}$ is block-diagonal, we have from \cite{sayed2012diffusion} that:
\begin{equation}
\rho({\cal B})=\rho\Big({\cal A}^T_2(I-{\cal M} {\cal R}){\cal A}^T_1\Big)\leq \rho\big(I-{\cal M} {\cal R}\big)
\label{eq.:spectral-inequalities}
\end{equation}
Therefore, if ${\cal R}$ is positive-definite, then choosing \mbox{$\mu_k<2/\lambda_{\max}(R_{k})$} ensures convergence of the algorithm in the mean so that $\E[\tilde{\w}_i]\rightarrow 0$ as $i\rightarrow\infty$.
However, when ${\cal R}$ is singular, it may hold that $\rho({\cal B})=1$, in which case choosing the step-sizes according to the above bound guarantees the boundedness of the mean error, $\E[\tilde{\w}_i]$, but not necessarily that it converges to zero. The following result clarifies these observations.
%---------------------------------------------%
%\theoremstyle{plain}
%\newtheorem{lem}{Lemma}
\newtheorem{thm}{Theorem}
\begin{thm}
\label{lemm.:rank_deficient_distributed_lms_thm1}
If the step-sizes are chosen to satisfy
\begin{equation}
    0<\mu_k<\frac{2} {\lambda_{\max}(R_k)}
    \label{eq.:step_size_difusion_lms_space_varying}
\end{equation}
then, under Assumption \ref{assm.:regressor assumption}, the diffusion algorithm is stable in the mean in the following sense:
(a) If $\rho({\cal B})<1$, then $\E[\tilde{\w}_i]$ converges to zero and (b) if $\rho({\cal B})=1$ then
\begin{align}
\lim_{i \rightarrow \infty}\big\|\E[{\tilde  \w}_i]\big\|_{b,\infty}&\leq\|I-\Ind(\Lambda)\|_{b,\infty} \,  \big \|\E[\tilde{\w}_{-1}]\big \|_{b,\infty}
\label{eq.:mean_perfomance5}
\end{align}
where $\|\cdot\|_{b,\infty}$ stands for the block-maximum norm, as defined in \cite{sayed2012diffusion,takahashi2010diffusion}.
\end{thm}
\begin{proof}
See Appendix \ref{apex.:Mean Convergence Proof}.
\end{proof}
%-----------------------------------------------------------%
In what follows, we examine recursion (\ref{eq.:network_vector_update}) and derive an expression for the asymptotic value of $\E[\w_i]$---see (\ref{eq.:lim-bomegai2}) further ahead. Before doing so, we first comment on a special case of interest, namely, result (\ref{eq.:network_global_decoupled_solutionA1A2I_PSD}) below.

{\it Special case}: Consider a network with ${A}_1={A}_2=I$ and an arbitrary right stochastic matrix $C$ satisfying (\ref{eq.sthocastic_matrix_c_conditions}).
Using (\ref{eq.:measurement_linear_model2}) and (\ref{eq.:localR})-(\ref{eq.:localr}), it can be verified that the following linear system of equations holds at each node $k$:
\be
R_k w^o = r_k
\label{eq.:NodeNormalEquation}
\ee
We show in Appendix \ref{apex.:error-bound-A1A2I} that under condition (\ref{eq.:step_size_difusion_lms_space_varying}) the mean estimate of the diffusion LMS algorithm at each node $k$ will converge to:
\begin{align}
\lim_{i\rightarrow \infty}\E[\w_{k,i}]=R_k^{\dag} r_k+\sum_{n={L_k+1}}^{MN_b} q_{k,n}q_{k,n}^T \E[\w_{k,-1}]
\label{eq.:network_global_decoupled_solutionA1A2I_PSD}
\end{align}
where $R_k^{\dag}$ represents the pseudo-inverse of $R_k$, and $\w_{k,-1}$ is the node initial value. This result is consistent with the mean estimate of the stand-alone LMS filter with rank-deficient input data (which corresponds to the situation $A_1=A_2=C=I$)\cite{mclernon2009convergence}.
Note that $R_k^{\dag} r_k$ in (\ref{eq.:network_global_decoupled_solutionA1A2I_PSD}) corresponds to the minimum-norm solution of $R_k w=r_k$. Therefore, the second term on the right hand side of (\ref{eq.:network_global_decoupled_solutionA1A2I_PSD}) is the deviation of the node estimate from this minimum-norm solution. The presence of this term after convergence is due to the zero eigenvalues  of $R_k$. If $R_k$ were full-rank so that $L_k=MN_b$, then this term would disappear and the node estimate will converge, in the mean, to its optimal value, $w^o$.  We point out that even though the matrices $\bar{R}_{u,\ell}$ are rank deficient since $N_b>1$, it is still possible for the matrices $R_k$ to be full rank owing to the linear combination operation in (\ref{eq.:localR}). This illustrates one of the benefits of employing the right-stochastic matrix $C$. However, if despite using $C$, $R_k$ still remains rank-deficient, the second term on the right-hand side of (\ref{eq.:network_global_decoupled_solutionA1A2I_PSD})
can be annihilated by proper node initialization (e.g., by setting $\E[\w_{k,-1}]=0$).
By doing so, the mean estimate of each node will then approach the unique minimum-norm solution, $R_k^{\dag} r_k$.

%------------------------------------------------------------------------------------------%
{\it General case}: Let us now find the mean estimate of the network for arbitrary left-stochastic matrices $A_1$ and $A_2$. Considering  definitions (\ref{eq.:Rcal_rank_defficientDiff})-(\ref{eq.:Rdcal_rank_defficientDiff})
and relation (\ref{eq.:NodeNormalEquation}) and noting that ${\cal A}^T_1 ({\mathbb 1} \otimes w^o)= {\cal A}^T_2 ({\mathbb 1} \otimes w^o)=({\mathbb 1} \otimes w^o)$, it can be verified that  $({\mathbb 1} \otimes w^o)$ satisfies the following linear system of equations:
\begin{align}
(I-{\cal B}) ({\mathbb 1} \otimes w^o) ={\mathcal A}^T_2{\mathcal M}r
\label{eq.:network_estimate_rank_def1}
\end{align}
This is a useful intermediate result that will be applied in our argument.

Next, if we iterate recursion (\ref{eq.:network_vector_update}) and apply the expectation operator, we then obtain
\be
\E [\w_i]={\mathcal{B}}^{i+1} \E[\w_{-1}]+\sum_{j=0}^{i}{\cal B}^j{ \cal A}^T_2{\mathcal M}r
\label{eq.:w_iteratedSolutionRank-deficient}
\ee
The mean estimate of the network can be found by computing the limit of this expression for $i\rightarrow \infty$. To find the limit of the first term on the right hand side of (\ref{eq.:w_iteratedSolutionRank-deficient}), we evaluate $\lim_{i\rightarrow \infty}{\mathcal{B}}^{i}$ and find conditions under which it converges. For this purpose, we introduce the Jordan decomposition of matrix ${\cal B}$ as \cite{meyer2000matrix}:
\begin{align}
{\cal B}&={\cal Z}\Gamma{\cal Z}^{-1}
\label{eq.:jordan-decompostion}
\end{align}
where ${\cal Z}$ is an invertible matrix, and $\Gamma$ is a block diagonal matrix of the form
\begin{align}
\Gamma=\diag \Big \{ \Gamma_{1},\Gamma_{2},\cdots,\Gamma_{s}\Big\}
\end{align}
where the $l$-th Jordan block, $\Gamma_{l}\in {\mathbb C}^{m_l \times m_l}$, can be expressed as:
\be
\Gamma_l=\gamma_l I_{m_l}+N_{m_l}
\label{eq.:block-decomposition}
\ee
In this relation, $N_{m_l}$ is some nilpotent matrix of size $m_l\times m_l$. Using decomposition (\ref{eq.:jordan-decompostion}), we can express ${\cal B}^i$ as
%\vspace{-0.1cm}
\begin{align}
{\mathcal{B}}^{i}&={\cal Z}\Gamma^i{\cal Z}^{-1}
\end{align}
Since $\Gamma$ is block diagonal, we have
%\vspace{-0.1cm}
\begin{align}
\Gamma^i=\diag \Big \{\Gamma_1^i,\Gamma_2^i,\cdots,\Gamma_s^i\Big\}
\end{align}
From this relation, it is deduced that $\lim_{i\rightarrow \infty} {\cal B}^i$ exists if $\lim_{i \rightarrow \infty} \Gamma_l^i$
exists for all $l\, \in\{1,\cdots,s\}$. Using (\ref{eq.:block-decomposition}), we can write \cite{meyer2000matrix}:
\begin{align}
\lim_{i\rightarrow \infty} \Gamma_l^i=\lim_{i\rightarrow \infty}  \gamma_l^{i-m_l}\Bigg(\gamma_l^{\scrsmall m_l} I_{ m_{\scrsmall l}}+\sum_{p=1}^{m_l-1} \dbinom{i}{p} \gamma_l^{m_l-p}N_{\scrsmall m_l}^p \Bigg)
\label{eq.:limGammaL}
\end{align}
When ${i\rightarrow \infty}$,  $\gamma_l^{i-m_l}$ becomes the dominant factor in this expression. Note that
under condition (\ref{eq.:step_size_difusion_lms_space_varying}), we have $\rho({\cal B})\leq 1$ which in turn implies that
the magnitude of the eigenvalues of $\cal B$ are bounded as $0 \leq |\gamma_n|\leq 1$. Without loss of generality, we assume that the eigenvalues of ${\cal B}$ are arranged as $|\gamma_1|\leq \cdots \leq |\gamma_{\scrsmall{L}}| < |\gamma_{\scrsmall{L+1}}|=\cdots=|\gamma_{ s}|=1$. Now we examine the limit (\ref{eq.:limGammaL}) for every $|\gamma_l|$ in this range. Clearly for $|\gamma_l|<1$, the limit is
zero (an obvious conclusion since in this case $\Gamma_l$ is a stable matrix). For $|\gamma_l|=1$, the limit is the identity matrix if $\gamma_l=1$ and $m_l=1$. However, the limit does not exist for unit magnitude complex eigenvalues and eigenvalues with value \mbox{-1}, even when $m_l=1$. Motivated by these observations, we introduce the following definition.

{\bf Definition}:
We refer to matrix ${\cal B}$ as \emph{power convergent} if (a) its eigenvalues $\gamma_n$ satisfy $0\leq |\gamma_n| \leq 1$, (b) its unit magnitude eigenvalues are all equal to one, and (c) its Jordan blocks associated with $\gamma_n=1$ are all of size $1\times 1$.
\hfill $\blacksquare$

\noindent {\em Example 1}: Assume $N_b=1$, $B_k=I_M$, and uniform step-sizes and covariance matrices across the agents, i.e.,  $\mu_k\equiv \mu$, $R_{u,k}\equiv R_u$ for all $k$. Assume further that $C$ is doubly-stochastic (i.e., $C^T {\mathbb 1}={\mathbb 1}=C{\mathbb 1})$ and $R_u$ is singular. Then, in this case, the matrix ${\cal B}$ can be written as the Kronecker product ${\cal B}=A_2^TA_1^T\otimes (I_M-\mu R_u)$. For strongly-connected networks where  $A_1A_2$ is a primitive matrix, it follows from the Perron-Frobenius Theorem \cite{horn2003matrix} that $A_1A_2$ has a single unit-magnitude eigenvalue at one, while all other eigenvalues have magnitude less than one. We conclude in this case, from the properties of Kronecker products and under condition (\ref{eq.:step_size_difusion_lms_space_varying}), that ${\cal B}$ is a power-convergent matrix. \hfill $\blacksquare$\\

\noindent {\em Example 2}: Assume $M=2$, $N=3$, $N_b=1$, $B_k=I_M$, and uniform step-sizes and covariance matrices across the agents again. Let $A_2=I=C$ and select
\be
A_1=A=\left [\begin{array}{cccc}1/2&0&0\\1/2&0&1\\0&1&0\end{array} \right ]
\ee
which is not primitive. Let further $R_u=\mbox{\rm diag}\{\beta,0\}$ denote a singular covariance matrix.  Then, it can be verified in this case  the corresponding matrix ${\cal B}$ will have an eigenvalue with value $-1$ and is not power convergent. \hfill $\blacksquare$

Returning to the above definition and assuming ${\cal B}$ is power convergent, then this means that the Jordan decomposition (\ref{eq.:jordan-decompostion}) can be rewritten as:
\be
{\cal B}=[\underbrace{{\cal Z}_1\; {\cal Z}_2}_{\cal Z}] \underbrace{\left[
                                   \begin{array}{cc}
                                      J& 0  \\
                                      0& I \\
                                   \end{array}
                                  \right]}_{\Gamma}
                                  \underbrace{
                                  \left[
                                    \begin{array}{c}
                                      \bar{\cal Z}_1 \\
                                      \bar{\cal Z}_2 \\
                                    \end{array}
                                  \right]}_{{\cal Z}^{-1}}
                                  \label{eq.:cal-B-definition}
\ee
where $J$ is a Jordan matrix with all eigenvalues strictly inside the unit circle, and the identity matrix  inside $\Gamma$ accounts for the eigenvalues with value one. In (\ref{eq.:cal-B-definition}) we further partition ${\cal Z}$ and ${\cal Z}^{-1}$ in accordance with the size of $J$. Using (\ref{eq.:cal-B-definition}), it is straightforward to verify that
\begin{align}
\lim_{i \rightarrow \infty} {\cal B}^{i+1}&={\cal Z}_2{\bar{\cal Z}}_2
\label{eq.:Bi-limit}
\end{align}
and if we multiply both sides of (\ref{eq.:network_estimate_rank_def1}) from the left by $\bar{\cal Z}_2$, it also follows that
\begin{align}
{\bar{\cal Z}}_2{ \cal A}^T_2{\mathcal M}r&=0
\label{eq.:barCalZ2-Mr}
\end{align}
Using these relations, we can now establish the following result, which describes the limiting behavior of the weight vector estimate.
%-------------------------------------------------------------------------------------------------------------------------------------------------------%
%\theoremstyle{plain}
%\newtheorem{thm}{Theorem}
\begin{thm}
\label{lemm.:mean estimate-general}
If the step-sizes $\{\mu_1,\cdots,\mu_N\}$ satisfy (\ref{eq.:step_size_difusion_lms_space_varying})  and matrix ${\cal B}$ is power convergent, then the mean estimate of the network given by (\ref{eq.:w_iteratedSolutionRank-deficient}) asymptotically converges to:
\begin{align}
\lim_{i\rightarrow \infty} &\E[\w_i]=({\cal Z}_2{\bar{\cal Z}}_2)\, \E[\w_{-1}]+(I-{\cal B})^{-} {\cal A}^T_2{\mathcal M}r
\label{eq.:lim-bomegai2}
\end{align}
where the notation $X^-$ denotes a (reflexive) generalized inverse for the matrix $X$. In this case,  the generalized inverse for $I-{\cal B}$ is given by
\be
(I-{\cal B})^{-} = {\cal Z}_1 (I-J)^{-1}\bar{{\cal Z}}_1
\ee
which is in terms of the factors $\{{\cal Z}_1,\bar{\cal Z}_1,J\}$ defined in (\ref{eq.:cal-B-definition}).
\end{thm}
\begin{proof}
See Appendix \ref{apex.:mean estimate-general}.
\end{proof}
%-----------------------------------------------------------%
We also argue in Appendix \ref{apex.:mean estimate-general} that the quantity on the right-hand side of (\ref{eq.:lim-bomegai2}) is invariant under basis transformations for the Jordan factors $\{{\cal Z}_1,\bar{\cal Z}_1,{\cal Z}_2,\bar{\cal Z}_2\}$. It can be verified that if $A_1=A_2=I$ then ${\cal B}$ will be symmetric and the result (\ref{eq.:lim-bomegai2}) will reduce to (\ref{eq.:network_global_decoupled_solutionA1A2I_PSD}). Now note that the first term on the right hand side of (\ref{eq.:lim-bomegai2}) is due to the zero eigenvalues of $I-{\cal B}$. From this expression, we observe that different initialization values generally lead to different estimates. However, if we set $\E[\w_{-1}]=0$, the algorithm converges to:
\be
\lim_{i\rightarrow \infty}\E[\w_i]=(I-{\cal B})^{-} {\cal A}^T_2{\mathcal M}r
\label{eq.:lim-bomegai3}
\ee
In other words, the diffusion LMS algorithm will converge on average to a generalized inverse solution of
the linear system of equations  defined by (\ref{eq.:network_estimate_rank_def1}).

When matrix ${\cal B}$ is stable so that $\rho({\cal B})<1$ then the factorization (\ref{eq.:cal-B-definition}) reduces to the form ${\cal B}={\cal Z}_1 J {\bar {\cal Z}}_1$ and $I-{\cal B}$ will be full-rank. In that case, the first term on the right hand side of (\ref{eq.:lim-bomegai2}) will be zero and the generalized inverse
will coincide with the actual matrix inverse so that (\ref{eq.:lim-bomegai2}) becomes
\begin{align}
\lim_{i\rightarrow \infty} &\E[\w_i]=(I-{\cal B})^{-1}{ \cal A}^T_2{\mathcal M}r
\label{eq.:lim-bomegaiIBfullRank2}
\end{align}
Comparing (\ref{eq.:lim-bomegaiIBfullRank2}) with (\ref{eq.:network_estimate_rank_def1}), we conclude that:
\begin{align}
\lim_{i\rightarrow \infty} &\E[\w_i]={\mathbb 1} \otimes w^o
\end{align}
which implies that the mean estimate of each node will be $w^o$. This result is in agreement with the previously developed mean-convergence analysis of diffusion LMS when the regression data have full rank covariance matrices \cite{sayed2012diffusion}.

%-----------------------Subsection----------------------------------------%
\subsection{Mean-Square Error Convergence}
\label{subsec.:Mean-SquareStability}
We now examine the mean-square stability of the error recursion (\ref{eq.:global_error_vector_w}) in the rank-deficient scenario. We begin by deriving an error variance relation as in \cite{sayed2008,al2003transient}. To find this relation, we form the weighted square ``norm'' of (\ref{eq.:global_error_vector_w}), and compute its expectation to obtain:
\begin{align}
\E\|\tilde \w_i\|^2_{\Sigma}=\E\big( \|\tilde \w_{i-1}&\|^2_{{\boldsymbol{\Sigma}}'}\big)+\E[\boldsymbol{g}^T_i {\mathcal M} {\mathcal A_2}\Sigma{\mathcal A}^T_2{\mathcal M}\boldsymbol{g}_i]
\label{variance_relation_1}
\end{align}
where $\|x\|^2_{\Sigma}=x^T \Sigma x$ and $\Sigma\geq0$ is  an arbitrary weighting matrix of compatible dimension that we are free to choose. In this expression,
\begin{align}
{\boldsymbol{\Sigma}}'={{\mathcal A_1}(I-{\mathcal M}\boldsymbol{\cal R}_i)^T {\mathcal
A_2}\Sigma{\mathcal A}^T_2(I-{\mathcal M}\boldsymbol{\cal R}_i){\mathcal A}^T_1}
\label{eq.SigmaBold}
\end{align}
Under the temporal and spatial independence conditions on the regression data from Assumption \ref{assm.:regressor assumption}, we can write:
\begin{align}
\E\big(\|\tilde \w_{i-1}&\|^2_{{\boldsymbol{\Sigma}}'}\big)=\E\|\tilde \w_{i-1}\|^2_{\E[{\boldsymbol{\Sigma}}']}
\end{align}
so that (\ref{variance_relation_1}) becomes:
\begin{align}
\E\|\tilde \w_i\|^2_{\Sigma}=\E\|\tilde \w_{i-1}&\|^2_{{\Sigma}'}+\Tr[\Sigma {\mathcal A}^T_2 {\mathcal M}{\mathcal
G}{\mathcal M}{\mathcal A_2}]
\label{variance_relation_2}
\end{align}
where  ${\mathcal G} \triangleq
\E[{\boldsymbol g}_i{\boldsymbol g}^T_i]$ is given by
\begin{align}
{\mathcal G}={\mathcal C}^T \text{diag} \big \{\sigma^2_{v,1}\bar{R}_{u,1},\ldots,\sigma^2_{v,N}{\bar R}_{u,N}\big \}{\mathcal C}
\label{eq.cal-G}
\end{align}
and
\begin{align}
\Sigma'&\triangleq \E [\boldsymbol{\Sigma}']={\cal B}^T\Sigma {\cal B}+O({\cal M}^2)
\approx{\cal B}^T\Sigma {\cal B}
\label{eq.:Sigma'Aprox}
\end{align}
We shall employ (\ref{eq.:Sigma'Aprox}) under the assumption of sufficiently small step-sizes where terms that depend on higher-order powers of the step-sizes are ignored. We next introduce
\be
{\cal Y} \triangleq {\cal A}^T_2{\cal M}{\cal G} {\cal M}{\cal A}_2
\label{eq.:calY}
\ee
and use (\ref{variance_relation_2}) to write:
\begin{align}
\E\|{\tilde \w}_i\|^2_{ \Sigma}=\E\|{\tilde \w}_{i-1}\|^2_{\Sigma'}+ \Tr(\Sigma {\cal Y})
\label{eq.:MSE-stability-analysis-1}
\end{align}
From (\ref{eq.:MSE-stability-analysis-1}), we arrive at
\small
\begin{align}
\E\|{\tilde \w}_i\|_{\Sigma}^2 =&\E\|\tilde{\w}_{-1}\|^2_{({\cal B}^T)^{i+1}\Sigma {\cal B}^{i+1}}+ \sum_{j=0}^{i}\Tr\Big(({\cal B}^T)^j \Sigma {\cal B}^j {\cal Y}\Big)
\label{eq.:MSE-stability-analysis-2}
\end{align}
\normalsize
To prove the convergence and stability of the algorithm in the mean-square sense, we examine the convergence of the terms on the right hand side of (\ref{eq.:MSE-stability-analysis-2}).

In a manner similar to (\ref{eq.:barCalZ2-Mr}), it is shown in Appendix \ref{apex.:mean-square-derivation} that the following property holds:
\be
\bar{\cal Z}_2 {\cal Y}=0,\;\;\;\;\;
{\cal Y}\bar{\cal Z}_2^T=0
\ee
Exploiting this result, we can arrive at  the following statement, which  establishes that relation  (\ref{eq.:MSE-stability-analysis-2}) converges as $i\rightarrow\infty$ and determines its limiting value.
%---------------------------------------------%
%\theoremstyle{plain}
%\newtheorem{thm}{Theorem}
\begin{thm}
\label{lemm.:mean-square-derivation}
Assume the step-sizes are sufficiently small and satisfy (\ref{eq.:step_size_difusion_lms_space_varying}). Assume also that ${\cal B}$ is power convergent. Under these conditions, relation (\ref{eq.:MSE-stability-analysis-2}) converges to
\begin{align}
\lim_{i \rightarrow \infty }\E\|{\tilde \w}_i\|^2_{\Sigma}&=\E\|\tilde{\w}_{-1}\|^2_{({\cal Z}_2\bar{\cal Z}_2)^T \Sigma {\cal Z}_2\bar{\cal Z}_2} \nonumber \\
&\qquad +\big(\vec({\cal Y})\big)^T (I-{\cal F})^{-1} \vec(\Sigma)
\label{eq.:network-steady-state-mean-square}
\end{align}
where
\be
{\cal F}\triangleq \big(({\cal Z}_1\otimes {\cal Z}_1)(J\otimes J)( {\bar{\cal Z}_1}\otimes {\bar{\cal Z}_1})\big)^T
\ee
and factors $\{{\cal Z}_1,\bar{\cal Z}_1,J\}$ are defined in  (\ref{eq.:cal-B-definition}).
\end{thm}
\begin{proof}
See Appendix \ref{apex.:mean-square-derivation}.
\end{proof}
In a manner similar to the proof at the end of Appendix \ref{apex.:mean estimate-general}, the term on the right hand side of (\ref{eq.:network-steady-state-mean-square}) is invariant under basis transformations on the factors $\{{\cal Z}_1,\bar{\cal Z}_1,{\cal Z}_2,\bar{\cal Z}_2\}$. Note that the first term on the right hand side of (\ref{eq.:network-steady-state-mean-square})  is the network penalty due to rank-deficiency. When the node covariance matrices are full rank, then choosing step-sizes according to (\ref{eq.:step_size_difusion_lms_space_varying}) leads to $\rho({\cal B})<1$. When this holds, then  ${\cal B}={\cal Z}_1J \bar{\cal Z}_1$. In this case, the first term on the right hand side of  (\ref{eq.:network-steady-state-mean-square}) will be zero, and ${\cal F}=({\cal B}\otimes{\cal B})^T$. In this case, we obtain:
\begin{align}
\lim_{i \rightarrow \infty }\E\|{\tilde \w}_i\|^2_{\Sigma}=\big(\vec({\cal Y})\big)^T (I-{\cal F})^{-1} \vec(\Sigma)
\label{eq.:network-steady-state-mean-square-full-rank}
\end{align}
which is in agreement with the mean-square analysis of diffusion LMS strategies for regression data with full rank covariance matrices given in \cite{cattivelli2010diffusion,sayed2012diffusion}.
%-----------------------Subsection----------------------------------------%
\subsection{Learning Curves}
For each $k$, the MSD and EMSE measures are defined as:
\begin{align}
\eta_k&=\lim_{i \rightarrow \infty } \E\|\tilde \h_{k,i}\|^2=\lim_{i \rightarrow \infty }\E\|\tilde \w_{k,i}\|^2_{B_k^T B_k} \\
\zeta_k&=\lim_{i \rightarrow \infty } \E\|  \u_{k,i}\tilde \h_{k,i-1}\|^2=\lim_{i \rightarrow \infty } \E\|\tilde \w_{k,i-1}\|^2_{{\bar R}_{u,k}}
\end{align}
where $\tilde{\h}_{k,i}=h^o_k- \h_{k,i}$. These parameters can be computed from the network error vector (\ref{eq.:network-steady-state-mean-square}) through proper selection of the weighting matrix $\Sigma$ as follows:
\begin{align}
\eta_k=\lim_{i \rightarrow \infty } \E\|{\tilde \w}_i\|^2_{\Sigma_{\msd_k}},
\quad \zeta_k=\lim_{i \rightarrow \infty } \E\|{\tilde \w}_{i-1}\|^2_{\Sigma_{\emse_k}},
\end{align}
where
\be
\Sigma_{\msd_k}=\diag(e_k)\otimes(B_k^T B_k), \; \Sigma_{\emse_k}=\diag(e_k) \otimes {\bar R}_{u,k}
\label{eq.:Sigma-MSD-EMSE}
\ee
and $\{e_k \}_{k=1}^N$ denote the vectors of a canonical basis set in $N$ dimensional space. The network MSD and EMSE measures are defined as
\begin{align}
\eta_{\net}=\frac{1}{N}\sum_{k=1}^N \eta_k, \qquad \zeta_{\net}=\frac{1}{N}\sum_{k=1}^N\zeta_k
\label{eq.:steady_emse_msd_network}
\end{align}
We can also define MSD and EMSE measures over time as
\begin{eqnarray}
\eta_k(i) &=& \E\|\tilde{\h}_{k,i}\|^2=\E\|\tilde{\w}_i\|^2_{\Sigma_{\msd_k}}\\
\zeta_k(i) &=& \E\|\u_{k,i}\tilde{\h}_{k,i-1}\|^2=
\E\|\tilde{\w}_{i-1}\|^2_{\Sigma_{\emse_k}}
\end{eqnarray}
Using (\ref{eq.:MSE-stability-analysis-2}), it can be verified that these measures evolve according to the following dynamics:
\begin{align}
&\eta_k(i)= \eta_k(i-1)-\|w^o\|_{{\cal H}^i(I- {\cal H}) \sigma_{\msd_k}}+ \alpha^T {\cal H}^i \sigma_{\msd_k}
\label{eq.:transient_msd_nodes}
\end{align}
\begin{align}
&\zeta_k(i)= \zeta_k(i-1)-\|w^o\|_{ {\cal H}^i(I- {\cal H}) \sigma_{\emse_k}}+ \alpha^T {\cal H}^{i} \sigma_{\emse_k}
\label{eq.:transient_emse_nodes}
\end{align}
where
\begin{eqnarray}
{\cal H}&=&({\cal B}\otimes{\cal B})^T \\
\alpha&=&\vec({\cal Y})\\
\sigma_{\msd_k}&=&\vec(\Sigma_{\msd_k})\\
\sigma_{\emse_k}&=&\vec(\Sigma_{\emse_k})
\end{eqnarray}
To obtain (\ref{eq.:transient_msd_nodes}) and (\ref{eq.:transient_emse_nodes}), we set $\E[\w_{k,-1}]=0$ for all $k$.
%------------------------------------------------------------------------------------------------------------------------------------------------------------%
%--------------------------------------------------------------------------Section5------------------------------------------------------------------------%
\section{Computer Experiments}
\label{sec.:results}
In this section, we examine the performance of the diffusion strategy (\ref{eq.diff-step1})-(\ref{eq.diff-step4}) and compare the simulation results with the analytical findings. In addition, we present a simulation example that shows the application of the proposed algorithm in the estimation of space-varying parameters for a physical phenomenon modeled by a PDE system over two spatial dimensions.

%-----------------------------------------------------Subsection----------------------------------%
\subsection{Performance of the Distributed Solution}
We consider a one-dimensional network topology, illustrated by Fig. \ref{fig.:network-topology-and-parameters}, with $L=1$ and equally spaced nodes along the $x$ direction. We choose $A_1$ as the identity matrix, and compute $A_2$ and $C$ based on the uniform combination and Metropolis rules \cite{cattivelli2010diffusion,sayed2012diffusion}, respectively. We choose $M=2$ and $N_b=5$ and generate the unknown global parameter $w^o$ randomly for each experiment. We obtain $B_k$ using the shifted Chebyshev polynomials given by (\ref{eq.:ShiftedCountinuesChebyshev}) and compute the space varying parameters $h^o_k$ according to (\ref{eq.:local_global_relation}). The measurement data $\d_k(i),\, k \in \{1,2,\cdots,N\}$ are generated using the regression model (\ref{eq.:state_dependent_regression}). The SNR for each node $k$ is computed as
$\text{SNR}_k=\E\| \u_{k,i}{h^o_k}\|^2/\sigma^2_{v,k}$. The noise and the entries of the regression data are white Gaussian and satisfy Assumption \ref{assm.:regressor assumption}. The noise variances, $\{\sigma^2_{v,k}\}$, and the trace of the covariance matrices, $\{\Tr(R_{u,k})\}$, are uniformly distributed between $[0.05,0.1]$ and $[1,5]$, respectively.

Figure \ref{fig.:network-MSD-N=4} illustrates the simulation results for a network with $N=4$ nodes. For this experiment, we set $\mu_k=0.01$ for all $k$ and initialize each node at zero.
In the legend of the figure, we use the subscript $h$ to denote the MSD for ${\tilde \h}_{k,i}$ and the subscript $w$ to refer to the MSD of ${\tilde \w}_{k,i}$.  The simulation curves are obtained by averaging over $300$ independent runs. it can be seen that the simulated and theoretical results match well in all cases. To obtain the analytical results, we use expression (\ref{eq.:network-steady-state-mean-square}) to assess the steady-state values and expression (\ref{eq.:transient_msd_nodes}) to generate the theoretical learning curves.
%---------------------------------------------------------------------------------------------%
\begin{figure}
\centering
\subfigure[The network MSD.]{
\includegraphics[width=8cm,height=6cm]{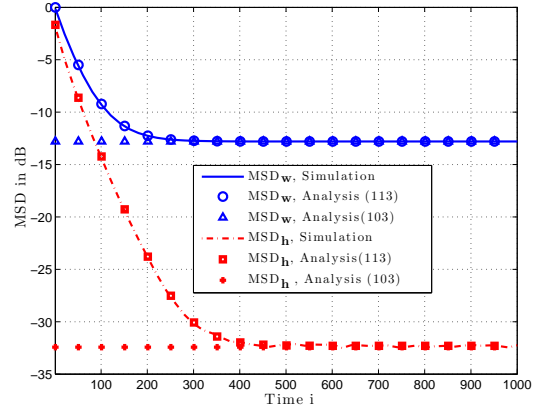}
\label{fig.:MSDTransientN4Nb5Network}
}
\subfigure[The MSD at some individual nodes.]{
\includegraphics[width=8cm,height=6cm]{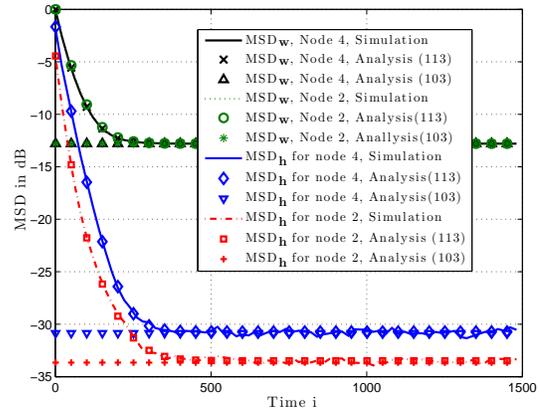}
\label{fig.:MSDTransientN4Nb5}
}
\caption{The network MSD learning curve for $N=4$.}
\label{fig.:network-MSD-N=4}
\end{figure}
%---------------------------------------------------------------------------------------------%

Two important points in Fig. \ref{fig.:network-MSD-N=4} need to be highlighted. First, note from the top plot  that the network MSD for ${\tilde \w}_{k,i}$ is larger than that for ${\tilde \h}_{k,i}$. This is because
\be
\E\|\tilde{\h}_{k,i}\|^2=\E\|\tilde{\w}_{k,i}\|^2_{B_k^T B_k}
\label{eq.:msdh-msdw-relation}
\ee
so that the MSD of ${\tilde \h}_{k,i}$ is a weighted version of the MSD of ${\tilde \w}_{k,i}$. In this experiment, the weighting leads to a lower estimation error. Second, note from the bottom plot that while the MSD values of $\tilde{\w}_{k,i}$ are largely independent of the node index, the same is not true for the MSD values of $\tilde{\h}_{k,i}$.
In previous studies on diffusion LMS strategies, it has been shown that, for strongly-connected networks, the network nodes approach a uniform MSD performance level \cite{sayed2013DiffusionMagazine}. The result in Fig. \ref{fig.:MSDTransientN4Nb5} supports this conclusion where it is seen that the MSD of ${\tilde \w}_{k,i}$ for nodes 2 and 4 converge to the same MSD level. However, note that the MSD of ${\tilde \h}_{k,i}$ is different for nodes 2 and 4. This difference in behavior is due to the difference in weighting across nodes from (\ref{eq.:msdh-msdw-relation}).

%-----------------------------------------%
\subsection{Comparison with Centralized Solution}
We next compare the performance of the diffusion strategy (\ref{eq.diff-step1})-(\ref{eq.diff-step4}) with the centralized solution (\ref{eq.:centralized-step1})--(\ref{eq.:centralized-step2}).
We consider a network with $N=10$ nodes with the topology illustrated by Fig. \ref{fig.:network-topology-and-parameters}. In this experiment, we set $\mu_k=0.02$ for all $k$, while the other network parameters are obtained following the same construction described for Fig. \ref{fig.:network-MSD-N=4}.  As the results in Fig. \ref{fig.:network-MSD-N=10} indicate,  the diffusion and centralized LMS solutions tend to the same MSD performance level in the $w$ domain. This conclusion is consistent with prior studies on the performance of diffusion strategies in the full-rank case over strongly-connected networks \cite{sayed2013DiffusionMagazine}. However, discrepancies in performance are seen between the distributed and centralized implementations in the $h$ domain, and the discrepancy tends to become larger for larger values of $N$.  This is because, in moving from the $w$ domain to the $h$ domain, the inherent aggregation of information that is performed by the centralized solution leads to enhanced estimates for the $h_{k}^o$ variables. For example, if the estimates $\w_{k,i}$ which are generated by the distributed solution are averaged prior to computing the $\h_{k,i}$, then it can be observed that the MSD values of ${\tilde \h}_{k,i}$  for both the centralized and the distributed solution will be similar.

%---------------------------------------------------------------------------------------------%
\begin{figure}
\centering
\subfigure[$N_b=5$.]{
\includegraphics[width=7.4cm,height=5.6cm]{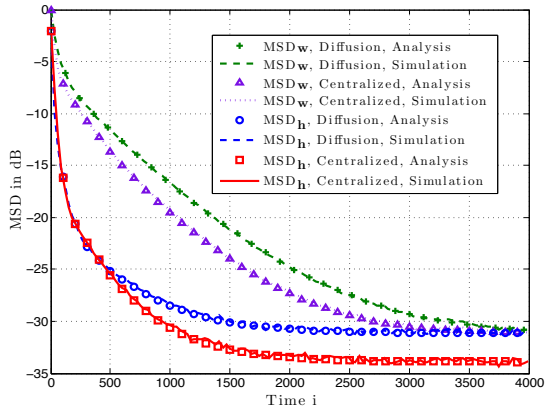}
\label{fig.:MSDTransientN10Nb5}
}
\subfigure[$N_b=10$.]{
\includegraphics[width=8.2cm,height=6.1cm]{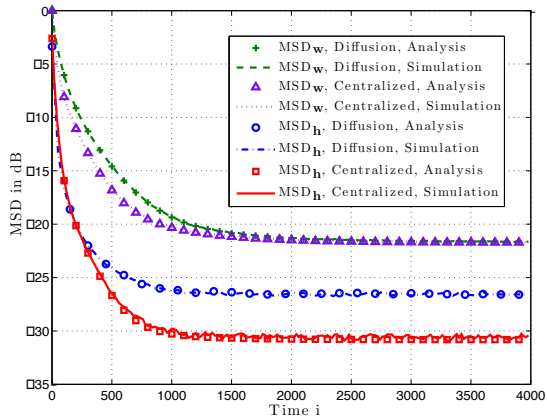}
\label{fig.:MSDTransientN10Nb10}
}
\caption{The network MSD learning curve for $N=10$.}
\label{fig.:network-MSD-N=10}
\end{figure}
%---------------------------------------------------------------------------------------------%

In these experiments, we also observe  that if we increase the number of basis functions, $N_b$, then both the centralized and diffusion algorithms will converge faster but their steady-state MSD performance will degrade. Therefore, in choosing the number of basis functions, $N_b$, there is a trade off between convergence speed and MSD performance.
%------------------------------------------------------Subsection-----------------------%
\subsection{Example: Two-Dimensional Process Estimation}
\label{subsec.:Diffusion LMS for Process Estimation}
In this example, we consider a two-dimensional network with $13 \times 13$ nodes that are equally spaced over the unit square $(x,y) \in [0, 1]\times[0, 1]$ with $\Delta x=\Delta y=1/12$ (see Fig. \ref{fig.:2DNetworkTopology}).  This network monitors a physical process $f(x,y)$ described by the Poisson PDE:
%------------------------------------------------------%
\begin{align}
\frac{\partial^2 f(x,y)}{\partial x^2}+\frac{\partial^2 f(x,y)}{\partial y^2}=h(x,y)
\label{eq.:TwoDimensionalPoissonProcess}
\end{align}
where $h(x,y): [0,\, 1]^2 \rightarrow \amsbb{R}$  is an unknown input function. The PDE satisfies the following boundary conditions:
\begin{align}
f(x,0)=f(0,y)=f(x,1)=f(1,y)=0 \nonumber
\end{align}
For this problem, the objective is to estimate $h(x,y)$, given noisy measurements collected by $N=N_x \times N_y=11 \times 11$ nodes corresponding to the {\it interior points} of the network. To discretize the PDE, we employ the finite difference method (FDM) with uniform spacing of $\Delta x$ and $\Delta y$. We define $x_{k_1}\triangleq k_1\Delta x$,\; $y_{k_2}\triangleq k_2 \Delta y$ and introduce the sampled values ${f}_{k_1,k_2}\triangleq f(x_{k_1} ,y_{k_2})$ and $h^o_{k_1,k_2}\triangleq h(x_{k_1},y_{k_2})$. We use the central difference scheme \cite{thomas1995numerical} to approximate the
second order partial derivatives:
\begin{align}
&\frac{{\partial}^2 f(x,y,t)}{\partial x^2}\approx\frac{1}{\Delta x^2}[f_{k_1+1,k_2}-2f_{k_1,k_2}+f_{k_1-1,k_2}]  \\
&\frac{{\partial}^2 f(x,y,t)}{\partial y^2}\approx\frac{1}{\Delta y^2}[f_{k_1,k_2+1}-2f_{k_1,k_2}+f_{k_1,k_2-1}]
\label{eq.:space_deravative}
\end{align}
This leads to the following discretized input function:
\begin{align}
h^o_{k_1,k_2}=&\frac{1}{\Delta x^2}\big(f_{k_1+1,k_2}+f_{k_1,k_2+1}+f_{k_1-1,k_2}\nonumber \\
&\qquad+f_{k_1,k_2-1}-4f_{k_1,k_2}\big)
\label{eq.:TwoDimensionalPoissonProcess_discretized}
\end{align}
For this example, the unknown input process is
\be
h^o_{k_1,k_2}=e^{-\kappa\big((k_1-4)^2+(k_2-4)^2\big )}-5e^{-\kappa\big((k_1-8)^2+(k_2-8)^2\big )}+1
\ee
where $\kappa=(N_x-1)^2/4$.
%------------------------------------------------------%
\begin{figure}
\centering
\subfigure[Network topology.]{
\includegraphics[width=6cm,height=4.5cm]{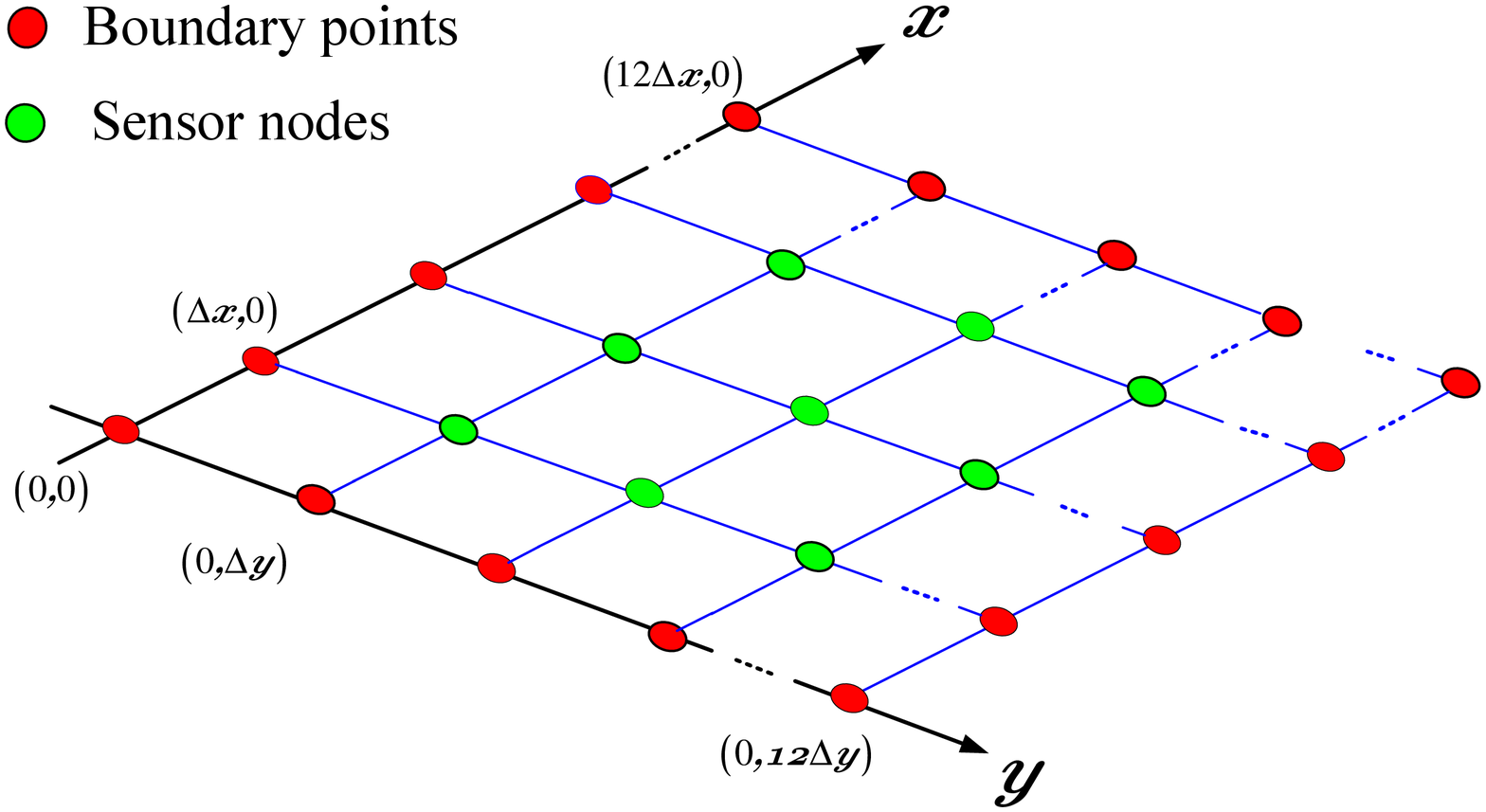}
\label{fig.:2DNetworkTopology}
}
\subfigure[$f_{k_1,k_2}$ over the space.]{
\includegraphics[width=6cm,height=4.5cm]{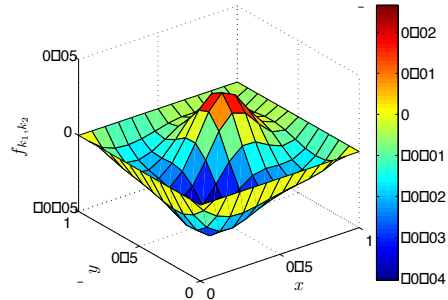}
\label{fig.:f_distribution}
}
\caption{Spatial distribution of $f(x,y)$ over the network grid $\{ (x_{k_1},y_{k_2}) \}$.}
\label{fig.:SourceEstimationExample}
\end{figure}
%------------------------------------------------------%
\begin{figure}
\centering
\includegraphics[width=4.5cm,height=3.4cm]{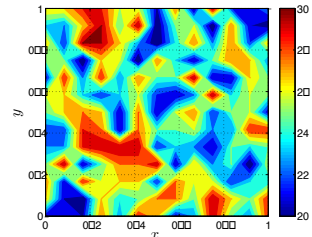}
\caption{Spatial distribution of SNR over the network.}
\label{fig.:networkSNR}
\end{figure}
%------------------------------------------------------%

To obtain $f_{k_1,k_2}$,  we solve (\ref{eq.:TwoDimensionalPoissonProcess})  using the Jacobi over-relaxation method \cite{bertsekas1989parallel}. Figure \ref{fig.:f_distribution} illustrates the values of $f_{k_1,k_2}$ over the spatial domain. For the estimation of $h_{k_1,k_2}$, the given information are the noisy measurement samples $\z_{k_1,k_2}(i)= f_{k_1,k_2}+{\n}_{k_1,k_2}(i)$. In this relation, the noise process ${\n}_{k_1,k_2}(i)$ is zero mean, temporally white and independent over space. For this network, the two dimensional reference signal is the distorted version of $h^o_{k_1,k_2}$ which is represented by $\d_{k_1,k_2}(i)$. The reference signal is obtained from (\ref{eq.:TwoDimensionalPoissonProcess_discretized}) with $f_{k_1,k_2}$ replaced by their noisy measured samples $\z_{k_1,k_2}(i)$, i.e.,
\begin{align}
\d_{k_1,k_2}(i)=&\frac{1}{\Delta x^2} \Big(\z_{k_1+1,k_2}(i)+\z_{k_1,k_2+1}(i)+\z_{k_1-1,k_2}(i) \nonumber \\
&\qquad+\z_{k_1,k_2-1}(i)-4\z_{k_1,k_2}(i)\Big)
\label{eq.:reference-signal2d}
\end{align}
\par \noindent
According to (\ref{eq.:reference-signal2d}), the linear regression model for this problem takes the following form:
\begin{align}
\d_{k_1,k_2}(i)=&\u_{k_1,k_2}(i) h^o_{k_1,k_2}+{\v}_{k_1,k_2}(i)
\label{eq.:2D-linear-model}
\end{align}
where $\u_{k_1,k_2}(i)=1$.
Therefore, in this example, we are led to a linear model (\ref{eq.:2D-linear-model}) with {\em deterministic}  as opposed to random regression data. Although we only studied the case of random regression data in this article, this example is meant to illustrate that the diffusion strategy can still be applied to models involving deterministic data in a manner similar to \cite{cattivelli2008diffusion,cattivelli2011distributed}.

To represent $h^o_{k_1,k_2}$ as a space-invariant parameter vector, we use  two-dimensional shifted Chebyshev basis functions \cite{mukundan2001image}. Using this representation, $h^o_{k_1,k_2}$ can be expressed as:
\begin{align}
h^o_{k_1,k_2}=\sum_{n=1}^{N_b} w^o_{n}\, p_{n,k_1,k_2}
\label{eq.:2d_parameter_interpolation}
\end{align}
where each element of the two-dimensional basis set is:
\begin{align}
p_{n,k_1,k_2}=b_{n_1,k_1}b_{n_2,k_2}
\label{eq.:2d_chebyshev function}
\end{align}
where $\{b_{n_1,k_1}\}$ and $\{b_{n_2,k_2}\}$ are the one-dimensional shifted Chebyshev polynomials in the $x$ and $y$ directions, respectively--recall (\ref{eq.:b-nk}).
%------------------------------------------------------%

In the network, each interior node communicates with its four immediate neighbors. We use $A_1=I$ and compute $C$ and $A_2$ by using the Metropolis and relative degree rules \cite{lopes2008diffusion, cattivelli2010diffusion,sayed2012diffusion}. All nodes are initialized at zero and
$\mu_k=0.01$ for all $k$.
The signal-to-noise ratio (SNR) of the network is uniformly distributed in the range $[20,30]$dB and is shown in Fig. \ref{fig.:networkSNR}.

%------------------------------------------------------%
Figures \ref{fig.:trueSourceValue} and \ref{fig.:EstimatedSourceValue} show three dimensional views of the true and  estimated input process using the proposed diffusion LMS algorithm after $3000$ iterations. Figure \ref{fig.:SourceEstimationMSDPerformance} illustrates the MSD of the estimated source, i.e.,
\mbox{$\lim_{i \rightarrow \infty} \E\|h^o_{k_1,k_2}-\h_{k_1,k_2}(i)\|^2$}.
%------------------------------------------------------%
\begin{figure}
\centering
\subfigure[True parameters.]{
\includegraphics[width=4.1cm,height=3cm]{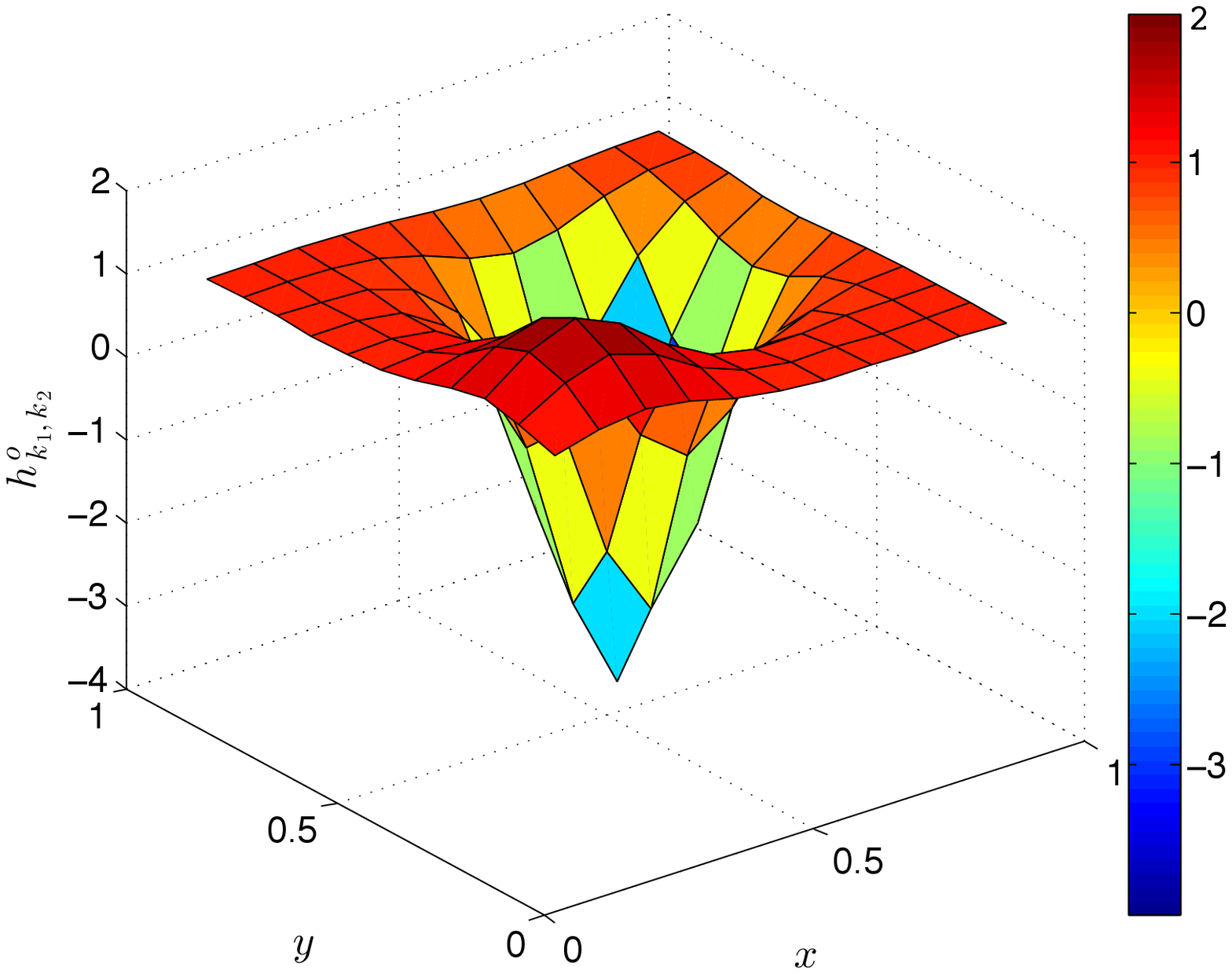}
\label{fig.:trueSourceValue}
}
\subfigure[Estimated parameters.]{
\includegraphics[width=4.1cm,height=3cm]{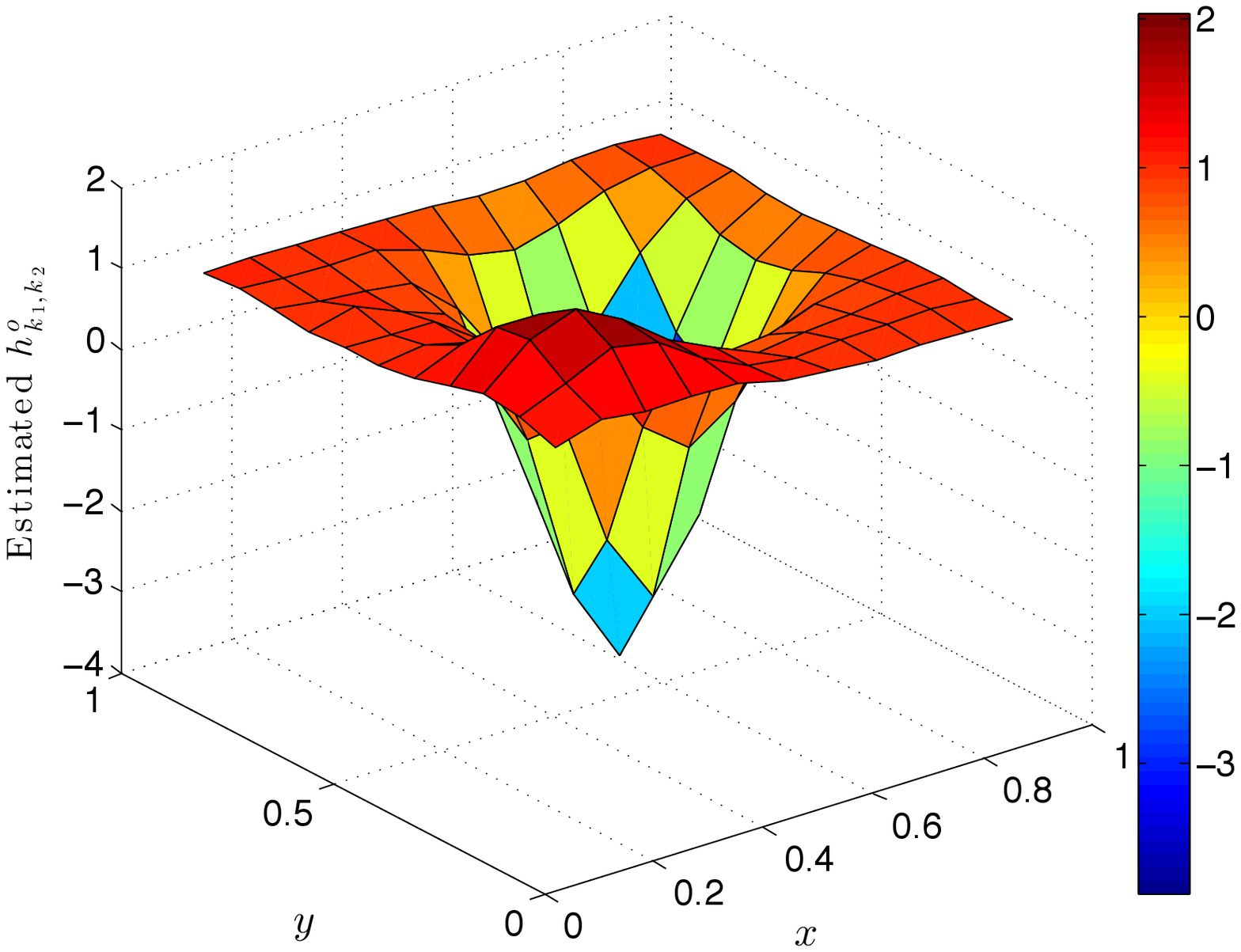}
\label{fig.:EstimatedSourceValue}
}
\caption{True and estimated  $h^o_{k_1,k_2}$ by diffusion LMS.}
\label{fig.:SourceAndEstimatedSourceValue}
\end{figure}
\begin{figure}
\centering
\includegraphics[width=7cm,height=5.5cm]{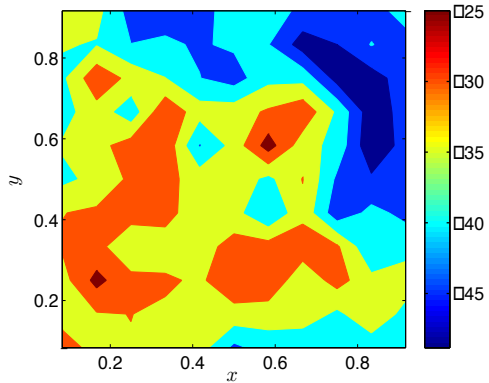}
\caption{Network steady-state MSD performance in dB.}
\label{fig.:SourceEstimationMSDPerformance}
\end{figure}
%------------------------------------------------------------------------------------------------------------------------------------------------------------%
%--------------------------------------------------------------------------Conclusion------------------------------------------------------------------------%
\section{Conclusion}
\label{sec.:conclusion}
By combining interpolation and distributed adaptive optimization, we proposed a diffusion LMS strategy for estimation and tracking of space-time varying parameters over networks. The proposed algorithm can find the space-varying parameters not only at the node locations but also at spaces where no measurement is collected. We showed that if the network experiences data with rank-deficient covariance matrices, the non-cooperative LMS algorithm will converge to different solutions at different nodes. In contrast, the diffusion LMS algorithm
is able to alleviate the rank-deficiency problem through its use of combination matrices
especially since, as shown by (\ref{eq.:spectral-inequalities}),  $\rho({\cal B})\leq \rho (I-{\cal M} {\cal R})$, where $I-{\cal M} {\cal R}$ is the coefficient matrix that governs the dynamics of the non-cooperative solution.
Nevertheless, if these mechanisms fail to mitigate the deleterious effect of the rank-deficient data, then the algorithm converges to a solution space where the error is bounded. We analyzed the performance of the algorithm in transient and steady-state regimes, and gave conditions under which the algorithm is stable in the mean and mean-square sense.
%------------------------------------------------------------------------------------------------------------------------------------------------------------%
%--------------------------------------------------------------------------Apendices-------------------------------------------------------------------------%
\begin{appendices}
%----------------------------------------------Appendix A------------------------------------%
\section{Mean Error Convergence}
\label{apex.:Mean Convergence Proof}
Based on the rank of ${\cal R}=\diag\{R_1,\cdots,R_N\}$, we have two possible cases:

{a) $R_k>0 \; \forall k \in \{1,\cdots,N\}$}:
As (\ref{eq.:mean_perfomance}) implies, $\E[\tilde  \w_i]$ converges to zero if $\rho({\cal B})<1$. In \cite{sayed2012diffusion}, it was shown that when ${\cal R}>0$, choosing the step-sizes according to (\ref{eq.:step_size_difusion_lms_space_varying}) guarantees $\rho({\cal B})<1$.

{ b) $\exists k \in \{1,\cdots,N\}$ for which $R_k$ is rank-deficient}: For this case, we first show that
\be
\big \|{\cal B}^{i+1} \big \|_{b,\infty}\leq \big \| \big(I-{\cal M}\Lambda)^{i+1} \big \|_{b,\infty}
\label{eq.:meanError-inequality1}
\ee
where $\|\cdot\|_{b,\infty}$ denotes the block-maximum norm for block vectors with block entries of size $MN_b \times 1$ and block matrices with blocks of size $MN_b\times MN_b$. To this end, we note that for the left-stochastic matrices $A_1$ and $A_2$, we have $\|{\cal A}_1^T\|_{b,\infty}=\|{\cal A}_2^T\|_{b,\infty}=1$ \cite{sayed2012diffusion}, and use the sub-multiplicative property of the block maximum norm \cite{zhao2012imperfect} to write:
\begin{align}
 \big \|{\cal B}^{i+1} \big \|_{b,\infty}&\leq \|{\cal A}_2^T\|_{b,\infty} \; \|I-{\cal M}{\cal R}\|_{b,\infty}\; \|{\cal A}_1^T\|_{b,\infty} \times \cdots \nonumber \\
 & \qquad \times \|{\cal A}_2^T\|_{b,\infty} \; \|I-{\cal M}{\cal R}\|_{b,\infty}\; \|{\cal A}_1^T\|_{b,\infty}
\nonumber \\
 &=\big \|I-{\cal M}{\cal R}\big \|^{i+1}_{b,\infty} \label{eq.:meanError-inequality2}
\end{align}
If we introduce the (block) eigendecomposition of ${\cal R}$ (\ref{eq.:EigenDecompositionOfD}) into (\ref{eq.:meanError-inequality2}) and consider the fact that the block-maximum norm is invariant under block-diagonal unitary matrix transformations \cite{sayed2012diffusion, takahashi2010diffusion}, then inequality (\ref{eq.:meanError-inequality2}) takes the form:
\begin{align}
\big \|{\cal B}^{i+1} \big \|_{b,\infty} & \leq \big \|I-{\mathcal M}{\Lambda}\big \|^{i+1}_{b,\infty}
\label{eq.:meanError-inequality3}
\end{align}
Using the property $\|X\|_{b,\infty}=\rho(X)$ for a block diagonal Hermitian matrix $X$ \cite{sayed2012diffusion}, we obtain:
\begin{align}
\big \|(I-{\cal M}{\Lambda})^{i+1} \big \|_{b,\infty}=&\rho\Big((I-{\cal M}\Lambda)^{i+1}\Big)  \nonumber \\
=&\max_{\substack{1 \leq k \leq N \\  1 \leq n \leq M N_b}}  \Big | \big(1-\mu_k \lambda_k(n)\big)^{i+1} \Big | \nonumber \\
=&\Big (\max_{\substack{1 \leq k \leq N \\  1 \leq n \leq M N_b}}  | 1-\mu_k \lambda_k(n)| \Big )^{i+1} \nonumber \\
=&\Big(\rho(I-{\cal M}{\Lambda})\Big)^{i+1} \nonumber \\
=&\big \|I-{\cal M}{\Lambda}\big \|^{i+1}_{b,\infty} \label{eq.:meanError-inequality4}
\end{align}
Using (\ref{eq.:meanError-inequality4}) in (\ref{eq.:meanError-inequality3}), we arrive at (\ref{eq.:meanError-inequality1}). We now proceed to show the boundedness of the mean error for case (b). We iterate (\ref{eq.:mean_perfomance}) to get:
\begin{align}
\E[\tilde  \w_i]={\cal B}^{i+1} \E[\tilde{\w}_{-1}]
\label{eq.:mean_perfomance2}
\end{align}
Applying the block maximum norm to (\ref{eq.:mean_perfomance2}) and using inequality (\ref{eq.:meanError-inequality1}), we obtain:
\begin{align}
\lim_{i \rightarrow \infty} \big \| \E[\tilde{\w}_i] \big \|_{b,\infty}&\leq \lim_{i \rightarrow \infty} \big \| (I-{\cal M} \Lambda)^{i+1} \big \|_{b,\infty} \; \big \|\E[\tilde{\w}_{-1}]\big \|_{b,\infty}
\label{eq.:mean_perfomance3}
\end{align}
The value of $\lim_{i\rightarrow \infty} \|(I-{\mathcal M}{\Lambda})^{i+1}\|_{b,\infty}$ can be computed by evaluating the limits of its diagonal entries. Considering the step-sizes as in (\ref{eq.:step_size_difusion_lms_space_varying}), the diagonal entries are computed as:
\be
\lim_{i\rightarrow \infty}\big(1-\mu_k \lambda_k(n)\big)^{i+1}=\left \{\begin{array}{l l}
1, \,\,\,& \text{if}\, \lambda_k(n)=0\\
0, \,\,\,& \text{otherwise}
\end{array}
\right.
\ee
Therefore, (\ref{eq.:mean_perfomance3}) reads as:
\begin{align}
\lim_{i \rightarrow \infty}\big \|\E[{\tilde  \w}_i]\big\|_{b,\infty}&\leq \|I-\Ind(\Lambda)\|_{b,\infty}
\; \big \|\E[\tilde{\w}_{-1}]\big \|_{b,\infty}
\label{eq.:mean_perfomanc6}
\end{align}

%----------------------------------------------------%
\section{Mean Behavior When ($A_1=A_2=I$)}
\label{apex.:error-bound-A1A2I}
Setting $A_1=A_2=I$ in the diffusion recursions (\ref{eq.diff-step1})-(\ref{eq.diff-step3}) and  subtracting $w^o$ from both sides of (\ref{eq.diff-step2}), we get:
\begin{equation}
\tilde  \w_{k,i}=\tilde  \w_{k,i-1}-\mu_k\sum_{\ell \in {\cal N}_k}c_{\ell,k} B^T_{\ell} \u^T_{\ell,i}(\d_{\ell}(i)-\u_{\ell,i}B_{\ell}  \w_{k,i-1})
\end{equation}
Under Assumption \ref{assm.:regressor assumption} and using $\d_{\ell}(i)=\u_{\ell,i}B_{\ell}w^o+\v_{\ell}(i)$, we obtain:
\begin{align}
\E[\tilde  \w_{k,i}]&=Q_k[I-\mu_k \Lambda_k]Q^T_k \, \E[\tilde  \w_{k,i-1}]
\end{align}
We define $ \p_{k,i} \triangleq {Q}_k^T \tilde \w_{k,i}$ and start from some initial condition to arrive at
\begin{align}
\E[\p_{k,i}]=[I-\mu_k \Lambda_k]\E[ \p_{k,i-1}]=[I-\mu_k \Lambda_k]^{i+1} \E[\p_{k,-1}] \nonumber
\end{align}
If we choose the step-sizes according to (\ref{eq.:step_size_difusion_lms_space_varying}) then we get:
\begin{align}
\lim_{i\rightarrow \infty} \E[ \p_{k,i}]=\big[I-\Ind(\Lambda_k)\big] \E [\p_{k,-1}]
\label{eq.:lms_rank_defficient_weight_error_vector}
\end{align}
Equivalently, this can be written as:
\begin{align}
\lim_{i\rightarrow \infty} \E[ {\tilde \w}_{k,i}]=Q_k \big[I-\Ind(\Lambda_k)\big] Q^T_k \, \E[{\tilde \w}_{k,-1}]
\label{eq.:lms_rank_defficient_weight_error_vector2}
\end{align}
This result indicates that the  mean error does not grow unbounded. Now from (\ref{eq.:NodeNormalEquation}), we can verify that:
\begin{equation}
Q_k \Ind(\Lambda_k) Q_k^T w^o = R_k^\dag r_k
\end{equation}
Then, upon substitution of $\tilde{\w}_{k,i}=w^o-\w_{k,i}$ into (\ref{eq.:lms_rank_defficient_weight_error_vector2}), we obtain:
\begin{align}
\lim_{i \rightarrow \infty} \E[\w_{k,i}]
&=\small Q_k \Ind(\Lambda_k) Q_k^T w^o+Q_k [I-\Ind(\Lambda_k)]Q_k^T \E[\w_{k,-1}] \normalsize \nonumber \\
&= R_k^{\dagger}r_k+\sum_{n=L_k+1}^{MN_b} q_{k,n} q_{k,n}^T \E[\w_{k,-1}]
\end{align}
%-------------------------Apendix C-----------------------------------------------%
\section{Proof of Lemma \ref{lemm.:mean estimate-general}}
\label{apex.:mean estimate-general}
From (\ref{eq.:Bi-limit}), we readily deduce that
\begin{align}
\lim_{i \rightarrow \infty} {\cal B}^{i+1}\E[\w_{-1}]=({\cal Z}_2{\bar{\cal Z}}_2)\, \E[\w_{-1}]
\label{eq.:limit-first-term}
\end{align}
On the other hand, from (\ref{eq.:cal-B-definition}), we have
\be
\lim_{i \rightarrow \infty} \sum_{j=0}^i {\cal B}^{j} {\cal A}^T_2{\cal M}r= \lim_{i \rightarrow \infty}\sum_{j=0}^i \big({\cal Z}_1 J^{j} \bar{\cal Z}_1+{\cal Z}_2 \bar{\cal Z}_2\big) {\cal A}^T_2{\cal M}r
\label{eq.:limit-second-term-v1}
\ee
Using (\ref{eq.:barCalZ2-Mr}), the term involving $\bar{\cal Z}_2$ cancels out and the above reduces to
\begin{align}
\lim_{i \rightarrow \infty}\sum_{j=0}^i {\cal B}^{j}{ \cal A}^T_2{\mathcal M}r&=\lim_{i \rightarrow \infty}\sum_{j=0}^i \big({\cal Z}_1 J^{j}{\bar{\cal Z}_1}\big)
{ \cal A}^T_2{\mathcal M}r \nonumber \\
&={\cal Z}_1(I-J)^{-1}{\bar{\cal Z}_1}{ \cal A}^T_2{\mathcal M}r
\label{eq.:limit-second-term-v1}
\end{align}
since $\rho(J)<1$. We now verify that the matrix
\be
X^{-}={\cal Z}_1(I-J)^{-1}{\bar{\cal Z}_1}
\ee
is a (reflexive) generalized inverse for the matrix $X=(I-{\cal B})$. Recall that a (reflexive) generalized inverse for a matrix $Y$ is any matrix $Y^{-}$ that satisfies the two conditions \cite{ben2003generalized}:
\begin{align}
YY^{-}Y&=Y
\label{eq.:generalized-inverse-condition-1}\\
Y^{-}YY^{-}&=Y^{-}
\label{eq.:generalized-inverse-condition-2}
\end{align}
To verify these conditions, we first note from ${\cal Z}{\cal Z}^{-1}=I$ and ${\cal Z}^{-1}{\cal Z}=I$ in (\ref{eq.:cal-B-definition}) that the following relations hold:
\begin{align}
&{\cal Z}_1{\bar {\cal Z}}_1+{\cal Z}_2{\bar {\cal Z}}_2=I
\label{eq.:Z-Prop-1}\\
&{\bar {\cal Z}}_1{\cal Z}_2=0 \label{eq.:Z-Prop-2} \\
&{\bar {\cal Z}}_2{\cal Z}_1=0 \label{eq.:Z-Prop-3}\\
&{\bar {\cal Z}}_1{\cal Z}_1=I  \label{eq.:Z-Prop-4}\\
&{\bar {\cal Z}}_2{\cal Z}_2=I \label{eq.:Z-Prop-5}
\end{align}
We further note that $X$ can be expressed as:
\be
X=(I-{\cal B})={\cal Z}_1(I-J){\bar{\cal Z}_1}
\ee
It is then easy to verify that the matrices $\{X,X^{-}\}$ satisfy conditions (\ref{eq.:generalized-inverse-condition-1}) and (\ref{eq.:generalized-inverse-condition-2}), as claimed.
Therefore, (\ref{eq.:limit-second-term-v1}) can be expressed as:
\be
\lim_{i \rightarrow \infty}\sum_{j=0}^i {\cal B}^{j}{ \cal A}^T_2{\mathcal M}r=(I-{\cal B})^{-}
{ \cal A}^T_2{\mathcal M}r
\label{eq.:limit-second-term-v2}
\ee
Substituting (\ref{eq.:limit-first-term}) and  (\ref{eq.:limit-second-term-v2}) into (\ref{eq.:w_iteratedSolutionRank-deficient})  leads to (\ref{eq.:lim-bomegai2}).
%------------------------------------------------------uniqness proof--------------------------------------%

Let us now verify that the right-hand side of (\ref{eq.:lim-bomegai2}) remains invariant under basis transformations for the Jordan factors $\{{\cal Z}_1,\bar{\cal Z}_1,{\cal Z}_2,\bar{\cal Z}_2\}$.
To begin with, the Jordan decomposition (\ref{eq.:cal-B-definition}) is not unique. Let us assume, however, that we fix the central term $\mbox{\rm diag}\{J,I\}$ to remain invariant and allow the Jordan factors  $\{{\cal Z}_1,\bar{\cal Z}_1,{\cal Z}_2,\bar{\cal Z}_2\}$ to vary. It follows from (\ref{eq.:cal-B-definition}) that
\be
\bar{\cal Z}_2{\cal B}=\bar{\cal Z}_2,\;\;\;\;{\cal B}{\cal Z}_2={\cal Z}_2
\ee
so that the columns of ${\cal Z}_2$ and the rows of $\bar{\cal Z}_2$ correspond to right and left-eigenvectors of ${\cal B}$, respectively, associated with the eigenvalues with value one. If we replace ${\cal Z}_2$ by any transformation of the form ${\cal Z}_2 {\cal X}_2$, where ${\cal X}_2$ is invertible, then by (\ref{eq.:Z-Prop-5}), $\bar{\cal Z}_2$ should be replaced by ${\cal X}_2^{-1}\bar{\cal Z}_2$. This conclusion can also be seen as follows. The new factor ${\cal Z}$ is given by

\small
\be
{\cal Z} \triangleq
\left [
\begin{array}{cc} {\cal Z}_1 & {\cal Z}_2{\cal X}_2
\end{array}
\right ]
=  \left [
\begin{array}{cc} {\cal Z}_1 & {\cal Z}_2
\end{array}
\right]
\left [
\begin{array}{cc}
I & 0\\ 0 & {\cal X}_2
\end{array}
\right]
\ee
\normalsize
and, hence, the new ${\cal Z}^{-1}$ becomes
\small
\be
{\cal Z}^{-1}= \left [
\begin{array}{c}
\bar{\cal Z}_1 \\ {\cal X}^{-1}_2 \bar{\cal Z}_2
\end{array}
\right ]
\ee
\normalsize
which confirms that $\bar{\cal Z}_2$ is replaced by ${\cal X}^{-1}_2\bar{\cal Z}_2$. It follows that the product ${\cal Z}_2\bar{\cal Z}_2$ remains invariant under arbitrary invertible transformations ${\cal X}_2$. Moreover, from (\ref{eq.:cal-B-definition}) we also have that
\be
\bar{\cal Z}_1{\cal B}=J\bar{\cal Z}_1,\;\;\;\;{\cal B}{\cal Z}_1={\cal Z}_1 J
\ee
Assume we replace ${\cal Z}_1$ by any transformation of the form ${\cal Z}_1 {\cal X}_1$, where ${\cal X}_1$ is invertible, then by (\ref{eq.:Z-Prop-4}), $\bar{\cal Z}_1$ should be replaced by ${\cal X}_1^{-1}\bar{\cal Z}_1$. However, since we want to maintain $J$ invariant, then this implies that the transformation ${\cal X}_1$ must also satisfy
\be
{\cal X}_1^{-1} J {\cal X}_1 = J
\ee
It follows that the product ${\cal Z}_1 (I-J)^{-1} \bar{\cal Z}_1$ remains invariant under such invertible transformations ${\cal X}_1$, since
\begin{eqnarray}
{\cal Z}_1 (I-J)^{-1} \bar{\cal Z}_1 &=&
{\cal Z}_1 {\cal X}_1{\cal X}_1^{-1}(I-J)^{-1}{\cal X}_1{\cal X}_1^{-1} \bar{\cal Z}_1\nonumber\\
&=&{\cal Z}_1 {\cal X}_1(I-{\cal X}_1^{-1} J {\cal X}_1)^{-1}{\cal X}_1^{-1} \bar{\cal Z}_1\nonumber\\
&=&{\cal Z}_1 {\cal X}_1(I-J)^{-1} {\cal X}_1^{-1}\bar{\cal Z}_1
\end{eqnarray}

%--------------------------------------------Apendix E-----------------------------------------------%
\section{PROOF OF LEMMA \ref{lemm.:mean-square-derivation}}
\label{apex.:mean-square-derivation}
We first establish that $\bar{\cal Z}_2{\cal Y}$ and ${\cal Y}\bar{\cal Z}^T_2$ are both equal to zero. Indeed, we start by replacing $r$ in (\ref{eq.:barCalZ2-Mr}) by its expression from (\ref{eq.:Rdcal_rank_defficientDiff}) and (\ref{eq.:localr}) as
$r={\cal C}^T \col\{\bar{r}_{du,1},\cdots,\bar{r}_{du,N}\}$ that leads to:
\be
\bar{\cal Z}_2{\cal A}_2^T {\cal M} {\cal C}^T \col\{\bar{r}_{du,1},\cdots,\bar{r}_{du,N}\} = 0
\ee
By further replacing $\bar{r}_{du,k}$ by their values from (\ref{eq.:bar_Ru}), we obtain:
\be
\bar{\cal Z}_2 {\cal A}_2^T {\cal M} {\cal C}^T \diag\{B_1^T,\cdots,B_N^T\} \col\{r_{du,1}, \cdots, r_{du,N}\}= 0
\label{eq.orthogonality-Z2mr}
\ee
This relation must hold regardless of the cross-correlation vectors $\{r_{du,k}\}$. Therefore,
\be
\bar{\cal Z}_2 {\cal A}_2^T {\cal M} {\cal C}^T \diag\{B_1^T,\cdots,B_N^T\}=0
\label{eq.orthogonality-Z2mr2}
\ee
We now define
\be
{\cal V}=\diag\{\sigma^2_{v,1}I_{MN_b},\cdots,\sigma^2_{v,N}I_{MN_b}\}
\ee
and rewrite expression (\ref{eq.:calY}) as
\begin{align}
{\cal Y}&= {\cal A}^T_2{\cal M} {\cal C}^T \diag\{B_1^T,\cdots,B_N^T\} \diag \{R_{u,1},\cdots,R_{u,N}\}\nonumber \\
&\qquad \times \diag\{B_1,\cdots,B_N\}\, {\cal V}\, {\cal C}{\cal M}{\cal A}_2
\label{eq.:calY-v2}
\end{align}
Multiplying this from the left by $\bar{\cal Z}_2$ and comparing the result with (\ref{eq.orthogonality-Z2mr2}), we conclude that
\be
\bar{\cal Z}_2{\cal Y} =0
\label{eq.orthogonality-Z2Y1}
\ee
Noting that ${\cal Y}$ is symmetric, we then obtain:
\be
{\cal Y} \bar{\cal Z}^T_2=0
\label{eq.orthogonality-Z2Y2}
\ee
Returning to recursion (\ref{eq.:MSE-stability-analysis-2}), we note first from (\ref{eq.:cal-B-definition}) that ${\cal B}$ can be rewritten as
\be
{\cal B}={\cal Z}_1 J {\bar{\cal Z}_1}+{\cal Z}_2{\bar{\cal Z}}_2
\label{eq.:cal-B-definition-2}
\ee
Since  ${\cal B}$ is power convergent, the first term on the right hand side of (\ref{eq.:MSE-stability-analysis-2}) converges to
\begin{align}
\lim_{i \rightarrow \infty}\E\|\tilde{\w}_{-1}&\|^2_{({\cal B}^T)^{i+1} \Sigma {\cal B}^{i+1}}=\E\|{\tilde \w}_{-1}\|^2_{({\cal Z}_2\bar{\cal Z}_2)^T \Sigma {\cal Z}_2\bar{\cal Z}_2}
\label{eq.:addtional-zero-initial-v2}
\end{align}
Substituting (\ref{eq.:cal-B-definition-2}) into the second term on the right hand side of (\ref{eq.:MSE-stability-analysis-2}) and using (\ref{eq.orthogonality-Z2Y1}) and (\ref{eq.orthogonality-Z2Y2}), we arrive at
\begin{align}
\lim_{i\rightarrow \infty}\sum_{j=0}^{i}\Tr\Big(({\cal B}^T)^j \Sigma {\cal B}^j {\cal Y}\Big)
=&\Tr\Big(\lim_{i\rightarrow \infty} \sum_{j=0}^{i}({\cal Z}_1 J^j {\bar{\cal Z}_1})^T \nonumber \\
&\times\Sigma ({\cal Z}_1 J^j {\bar{\cal Z}_1}){\cal Y}\Big)
\label{eq:mse-second-term-limit-v2}
\end{align}
If matrices $X_1$, $X_2$ and $\Sigma$ are of compatible dimensions, then the following relations hold\cite{sayed2012diffusion}:
\begin{align}
\Tr(X_1X_2)=\big(\vec (X_2^T)\big)^T \vec(X_1) \\
\vec(X_1\Sigma X_2)=(X_2^T\otimes X_1)\vec(\Sigma)
\end{align}
Using these relations in (\ref{eq:mse-second-term-limit-v2}), we obtain
\begin{align}
\Tr\Big(&\lim_{i\rightarrow \infty}\sum_{j=0}^{i}({\cal B}^T)^j \Sigma {\cal B}^j {\cal Y}\Big)=\Big(\vec({\cal Y}^T)\Big)^T \nonumber \\
&\times \Big(\lim_{i \rightarrow \infty} \sum_{j=0}^{i}({\cal Z}_1 J^j {\bar{\cal Z}_1})^T \otimes ({\cal Z}_1 J^j {\bar{\cal Z}_1})^T \Big)\vec(\Sigma)
\label{eq:mse-second-term-limit-v3}
\end{align}
This is equivalent to:
\begin{align}
\Tr\Big(&\sum_{j=0}^{\infty}({\cal B}^T)^j \Sigma {\cal B}^j {\cal Y}\Big)=\big(\vec({\cal Y})\big)^T \Big(\sum_{j=0}^{\infty}{\cal F}^j\Big)\vec(\Sigma)
\label{eq:mse-second-term-limit-v4}
\end{align}
where
\be
{\cal F}=\Big(({\cal Z}_1\otimes {\cal Z}_1)(J\otimes J)( {\bar{\cal Z}_1}\otimes {\bar{\cal Z}_1})\Big)^T
\ee
Since $\rho(J\otimes J)<1$, the series converges and we obtain:
\begin{align}
\Tr\Big(\lim_{i\rightarrow \infty}&\sum_{j=0}^{i}({\cal B}^T)^j \Sigma {\cal B}^j {\cal Y}\Big)=\Big(\vec({\cal Y})\Big)^T(I-{\cal F})^{-1}\vec(\Sigma)
\label{eq:mse-second-term-limit-v5}
\end{align}
Upon substitution of (\ref{eq.:addtional-zero-initial-v2}) and (\ref{eq:mse-second-term-limit-v5}) into (\ref{eq.:MSE-stability-analysis-2}), we arrive at (\ref{eq.:network-steady-state-mean-square}).
\end{appendices}
%------------------------------------------------------------------------------------------------------------------------------------------------------------%
%--------------------------------------------------------------------------References------------------------------------------------------------------------%
% Generated by IEEEtran.bst, version: 1.12 (2007/01/11)

%------------------------------------------------------------------------------------------------------------------------------------------------------------%
%--------------------------------------------------------------------------Biography-------------------------------------------------------------------------%
\vspace{-2cm}
\begin{IEEEbiography}[{\includegraphics[width=1in,height=1.25in,clip,keepaspectratio]{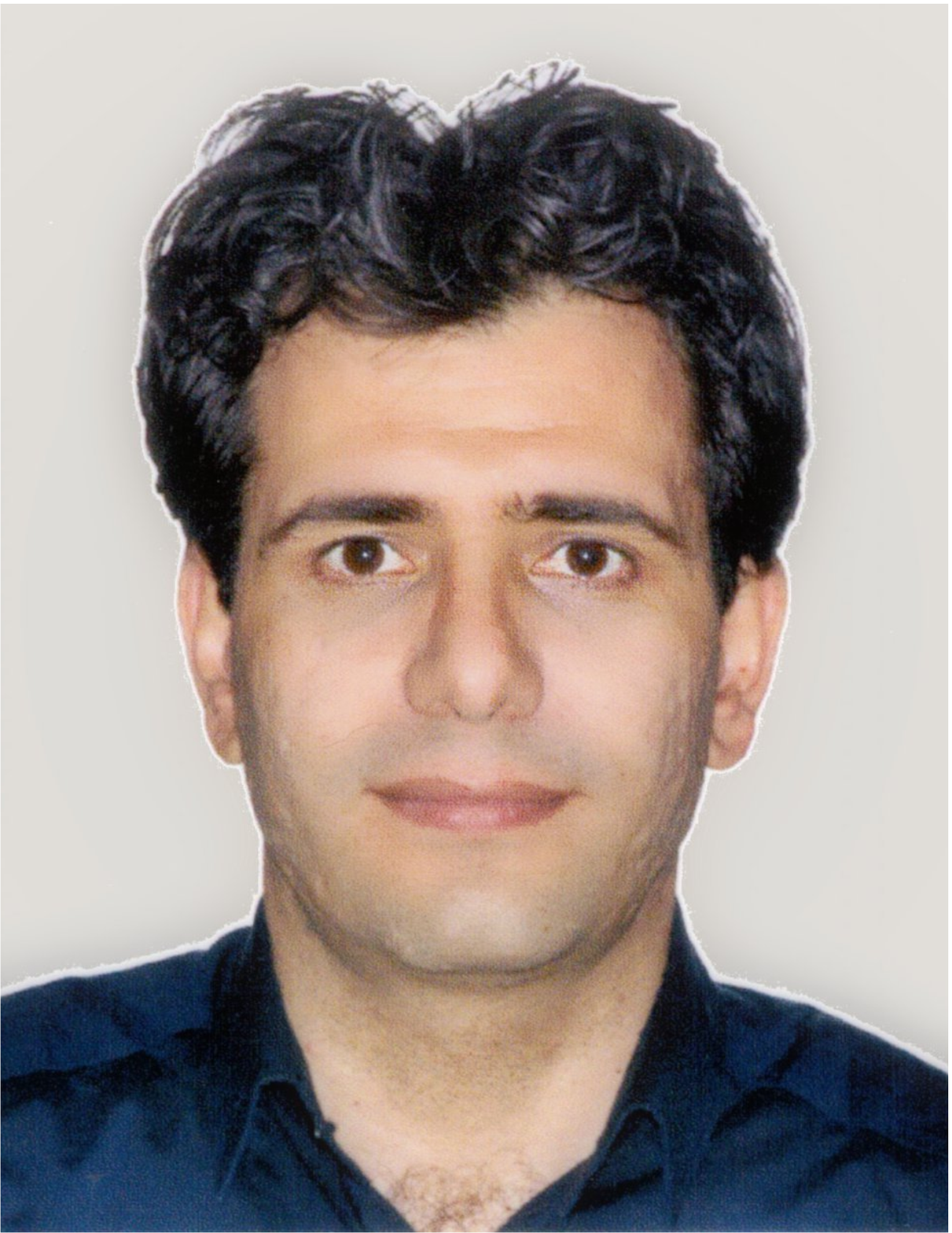}}]{Reza Abdolee} is currently a Ph.D. candidate at the Department of Electrical and Computer Engineering, McGill University, Montreal, Canada. In 2012, he was a research scholar at the Bell Labs, Alcatel-Lucent, Stuttgart, Germany. In 2011, he was a visiting Ph.D. student at the
Department of Electrical Engineering, University of California, Los Angeles (UCLA).  From 2006 to 2008, Mr. Abdolee worked as
a staff engineer at the Wireless Communication Center, University of Technology, Malaysia (UTM), where he implemented a switch-beam smart antenna system for wireless network applications. His research interests include communication theory, statistical signal processing, optimization, and hardware design and integration.  Mr. Abdolee was a recipient of several awards and scholarships, including, NSERC Postgraduate Scholarship, FQRNT Doctoral Research Scholarship, McGill Graduate Research Mobility Award, DAAD-RISE International Internship scholarship (Germany), FQRNT International Internship Scholarship, McGill Graduate Funding and Travel award, McGill International Doctoral Award, ReSMiQ International Doctoral Scholarship, Graduate Student Support Award (Concordia University), and International Tuition Fee Remission Award (Concordia University).
\end{IEEEbiography}
\begin{IEEEbiography}[{\includegraphics[width=1in,height=1.25in,clip,keepaspectratio]{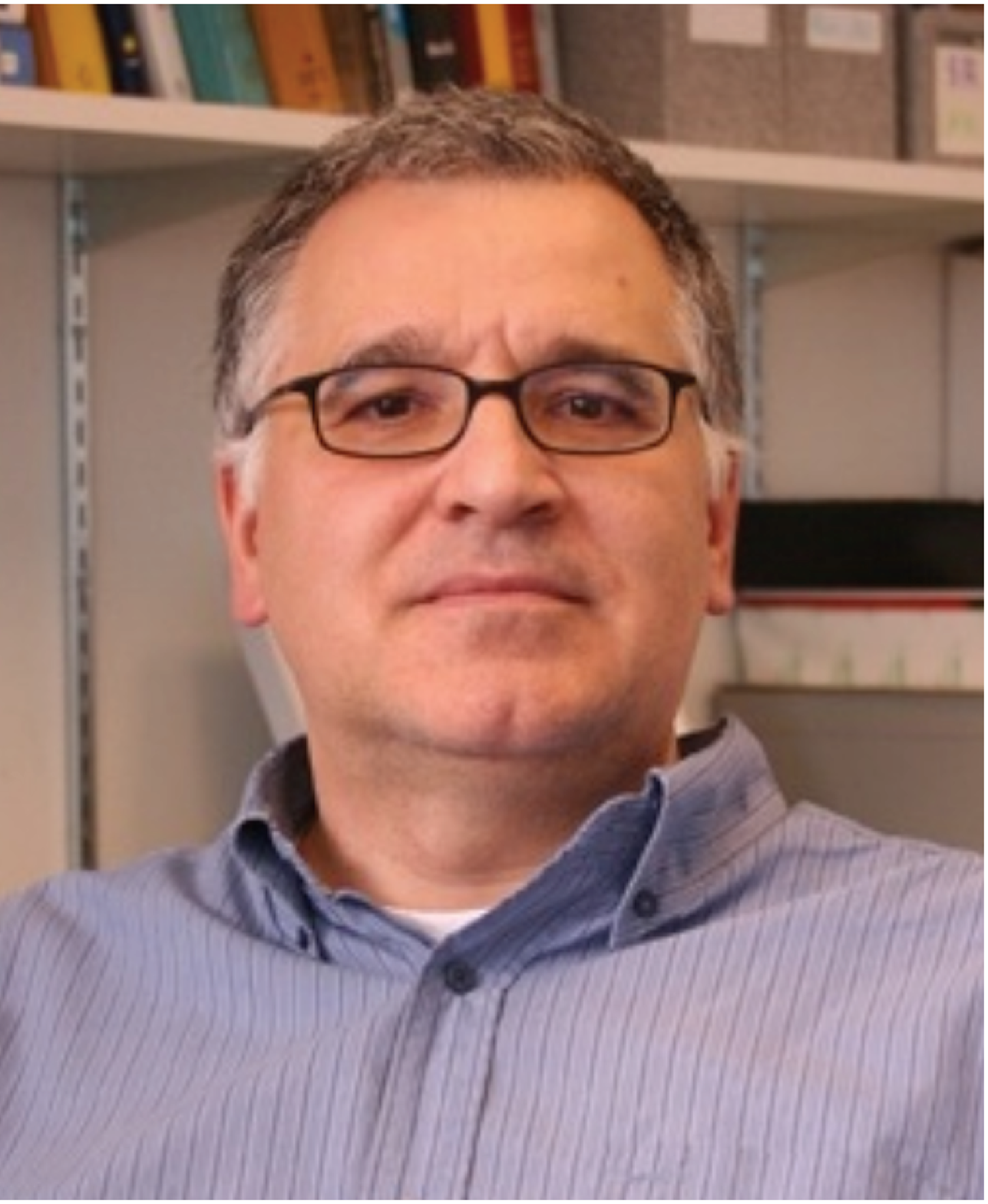}}]{Benoit Champagne} received the B.Ing. degree in Engineering Physics from the Ecole Polytechnique de Montréal in 1983, the M.Sc. degree in Physics from the Université de Montréal in 1985, and the Ph.D. degree in Electrical Engineering from the University of Toronto in 1990. From 1990 to 1999, he was an Assistant and then Associate Professor at INRS-Telecommunications, Université du Quebec, Montréal. In1999, he joined McGill University, Montreal, where he is now a Full Professor within the Department of Electrical and Computer Engineering. He also served as Associate Chairman of Graduate Studies in the Department from 2004 to 2007.

His research focuses on the development and performance analysis of advanced algorithms for the processing of information bearing signals by digital means. His interests span many areas of statistical signal processing, including detection and estimation, sensor array processing, adaptive filtering, and applications thereof to broadband communications and speech processing, where he has published extensively. His research has been funded by the Natural Sciences and Engineering Research Council (NSERC) of Canada, the Fonds de Recherche sur la Nature et les Technologies from the Government of Quebec, Prompt Quebec, as well as some major industrial sponsors, including Nortel Networks, Bell Canada, InterDigital and Microsemi.

He has been an Associate Editor for the IEEE Signal Processing Letters, the IEEE Trans. on Signal Processing and the EURASIP Journal on Applied Signal Processing. He has also served on the Technical Committees of several international conferences in the fields of communications and signal processing. He is currently a Senior Member of IEEE.
\end{IEEEbiography}
\vspace{-6cm}
\begin{IEEEbiography}[{\includegraphics[width=1in,height=1.25in,clip,keepaspectratio]{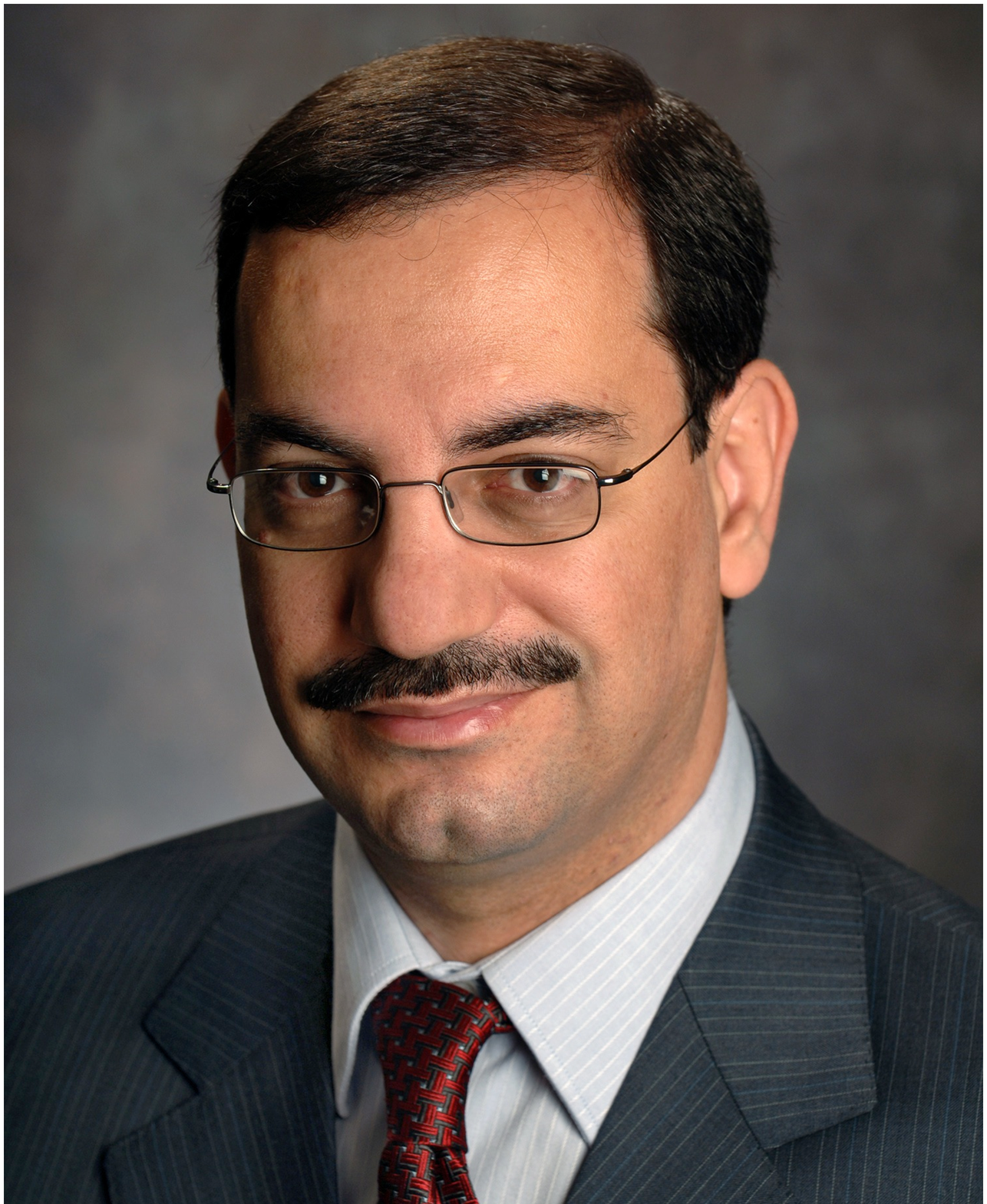}}]{Ali H. Sayed} is professor and former chairman of electrical engineering at the University of California, Los Angeles, where he directs the UCLA Adaptive Systems Laboratory. An author of over 430 scholarly publications and five books, his research involves several areas including adaptation and learning, information processing theories, statistical signal processing, network science, and biologically-inspired designs. His work has been recognized with several awards including the 2012 Technical Achievement Award from the IEEE Signal Processing Society, the 2005 Terman Award from the American Society for Engineering Education, a 2005 Distinguished Lecturer from the IEEE Signal Processing Society, the 2003 Kuwait Prize, and the 1996 IEEE Donald G. Fink Prize. He has also been awarded several paper awards from the IEEE Signal Processing Society (2002,2005,2012) and is a Fellow of both the IEEE and the American Association for the Advancement of Science (AAAS). He has been active in serving the Signal Processing community in various roles. Among other activities, he served as Editor-in-Chief of the IEEE Transactions on Signal Processing (2003-2005), Editor-in-Chief of the EURASIP J. on Advances in Signal Processing (2006-2007), General Chairman of ICASSP (Las, Vegas, 2008), and Vice-President of Publications of the IEEE Signal Processing Society (2009-2011). He also served as member of the Board of Governors (2007-2011), Awards Board (2005), Publications Board (2003-2005, 2009-2011), Conference Board (2007-2011), and Technical Directions Board (2008-2009) of the same Society and as
General Chairman of the IEEE International Conference on Acoustics, Speech,
and Signal Processing (ICASSP) 2008.
\end{IEEEbiography}
%------------------------------------------------------------------------------------------------------------------------------------------------------------%
\end{document}